\documentclass[12pt]{article}
\usepackage[utf8]{inputenc}
\usepackage[english]{babel}

\usepackage[numbers]{natbib}

\usepackage{amsthm}

\usepackage{amssymb}
\usepackage{amsmath}
\usepackage{mathrsfs}
\usepackage{url}
\usepackage{latexsym}
\usepackage{wasysym}
\usepackage{newlfont}
\usepackage{leqno}
\usepackage{xspace}
\usepackage{blkarray}

\usepackage{diagbox}
\usepackage{bbm}

\usepackage[prependcaption]{todonotes}

\usepackage{multicol}

\usepackage{eufrak}
\usepackage{soul}
\usepackage{enumerate}

\usepackage[capitalize]{cleveref}
\newtheorem{theorem}{Theorem}[section]

\newtheorem{lemma}[theorem]{Lemma}
\newtheorem{proposition}[theorem]{Proposition}
\newtheorem{corollary}[theorem]{Corollary}
\newtheorem{definition}[theorem]{Definition}

\newtheorem{observation}[theorem]{Observation}
\newtheorem{example}[theorem]{Example}

\crefname{equation}{eq.}{equations}
\crefname{definition}{def.}{definitions}
\crefname{section}{sec.}{sections}
\crefname{lemma}{lem.}{lemmata}
\crefname{corollary}{cor.}{corollaries}
\crefname{proposition}{prop.}{propositions}
\crefname{remark}{rem.}{remarks}
\crefname{enumi}{step}{steps}
\crefname{theorem}{thm.}{theorems}
\crefname{algorithm}{alg.}{algorithms}
\crefname{figure}{fig.}{figures}
\crefname{observation}{obs.}{observations}
\crefname{example}{ex.}{examples}

\usepackage[ruled]{algorithm}
\usepackage{algorithmic}

  \def\ZZ{{\mathbb Z}}

  \def\bm{\boldsymbol}

\def\Nz{\mathbb{N}_0}
\def\N{\mathbb{N}}
\def\Z{\mathbb{Z}}

\def\K{\mathbb{K}}
\def\Pr{\mathbb{P}}
\def\P{\mathcal{P}}
\def\A{\mathcal{A}}
\def\ev{\mathbbm{1}}
\def\MHB{\ensuremath{\mathtt{MHB}}\xspace}
\def\Res{\mathrm{Res}}
\def\Im{\mathrm{Im}}
\def\Ker{\mathrm{Ker}}

\newcommand{\mon}[1]{\bm{w^{#1}}}

\usepackage{multirow}

\PassOptionsToPackage{usenames}{color}
\usepackage{tikz}
\usetikzlibrary{matrix}

\usepackage{stmaryrd}

\newcommand\possiblebreak{\ifhmode\unskip\space\hfil\penalty0\hfilneg\fi}

\makeatletter
\newif\if@borderstar
\def\bordermatrix{\@ifnextchar*{%
  \@borderstartrue\@bordermatrix@i}{\@borderstarfalse\@bordermatrix@i*}%
}
\def\@bordermatrix@i*{\@ifnextchar[{%
  \@bordermatrix@ii}{\@bordermatrix@ii[()]}
}
\def\@bordermatrix@ii[#1]#2{%
  \begingroup
    \m@th\@tempdima8.75\p@\setbox\z@\vbox{%
      \def\cr{\crcr\noalign{\kern 2\p@\global\let\cr\endline }}%
      \ialign {$##$\hfil\kern 2\p@\kern\@tempdima & \thinspace %
      \hfil $##$\hfil && \quad\hfil $##$\hfil\crcr\omit\strut %
      \hfil\crcr\noalign{\kern -\baselineskip}#2\crcr\omit %
      \strut\cr}}%
    \setbox\tw@\vbox{\unvcopy\z@\global\setbox\@ne\lastbox}%
    \setbox\tw@\hbox{\unhbox\@ne\unskip\global\setbox\@ne\lastbox}%
    \setbox\tw@\hbox{%
      $\kern\wd\@ne\kern -\@tempdima\left\@firstoftwo#1%
        \if@borderstar\kern2pt\else\kern -\wd\@ne\fi%
      \global\setbox\@ne\vbox{\box\@ne\if@borderstar\else\kern 2\p@\fi}%
      \vcenter{\if@borderstar\else\kern -\ht\@ne\fi%
        \unvbox\z@\kern-\if@borderstar2\fi\baselineskip}%
        \if@borderstar\kern-2\@tempdima\kern2\p@\else\,\fi\right\@secondoftwo#1 $%
    }\null \;\vbox{\kern\ht\@ne\box\tw@}%
  \endgroup
}
\makeatother

\usepackage{soul}
  
\usepackage{amssymb,amsfonts,amsmath,xspace}

\begin{document}

\title{Bilinear systems with two supports: \\
Koszul resultant matrices, \\ eigenvalues, and eigenvectors 
}

\author{Mat\'{i}as R. Bender${}^{1}$, Jean-Charles Faug\`ere${}^{1}$, \\ Angelos Mantzaflaris${}^2$, and Elias Tsigaridas${}^1$}

\date{ ${}^1$ Sorbonne Universit\'e, \textsc{CNRS},
  \textsc{INRIA},\linebreak Laboratoire d'Informatique de Paris~6,
  \textsc{LIP6}, \linebreak \'Equipe \textsc{PolSys}, \linebreak 4
  place Jussieu, F-75005, Paris, France \linebreak
  \\
  ${}^2$ Johann Radon Institute for Computational and Applied
  Mathematics (RICAM), \linebreak Austrian Academy of Sciences, Linz,
  Austria \linebreak
  \\
  March 2018\footnote{This work is based on a paper originally
    published in ISSAC ’18, July 16–19, 2018, New York, USA \cite{bender_2bilinear_2018}} }

\maketitle

\begin{abstract}
  A fundamental problem in computational algebraic geometry is the computation
  of the resultant. A central question is when and how to compute it
  as the determinant of a matrix.
  whose elements are the coefficients of the input polynomials up-to sign.
  This problem is well understood for unmixed
  multihomogeneous systems, that is for systems consisting of multihomogeneous
  polynomials with the same support.
  However, little is known for mixed systems,
  that is for systems consisting of polynomials with different supports.

  We consider the computation of the multihomogeneous resultant of
  bilinear systems involving two different supports. We present a
  constructive approach that expresses the resultant as the exact
  determinant of a \emph{Koszul resultant matrix}, that is a matrix
  constructed from maps in the Koszul complex.
  We exploit the resultant matrix to propose an algorithm to solve
  such systems. In the process we extend the classical eigenvalues
  and eigenvectors criterion to a more general setting. Our extension
  of the eigenvalues criterion applies to a general class of matrices,
  including the Sylvester-type and the Koszul-type ones.
\end{abstract}









\section{Introduction} 
\label{sec:intro}
  
The resultant is a central object in elimination theory and
computational algebraic geometry. We use it to decide when an
overdetermined polynomial system has a solution and to solve
well-defined (square) systems. Moreover, it is one of the few tools
that take into account the sparsity of supports of the polynomials.

Usually, we compute the resultant as a quotient of determinants of two
matrices
\cite{macaulay1902some,jouanolou1997formes,dandrea2001explicit,dandrea_macaulay_2002}.
If we can compute the resultant as a determinant of only one matrix
whose non-zero entries are forms evaluated at the coefficients of the
input polynomials, 
then we have a \emph{determinantal formula}.
Among these cases, the best we can hope for is to have linear forms.
In general, determinantal formulas do not exist and it is an open
problem to decide when they do.

The matrices appearing in the computation of resultants have a strong
structure and we can classify them according to it.  For a system
$(f_0,\dots,f_n)$, a \emph{Sylvester-type formula} is a matrix that
represents a map $(g_0,\dots,g_n) \mapsto \sum_i g_i\, f_i$. It
extends the classical Sylvester matrix and it corresponds to the last
map of the Koszul complex of $(f_0,\dots,f_n)$. Another kind of
formula is the \emph{Koszul-type formula} that involves the other maps
of the Koszul complex. We call the matrices related to this formula
\emph{Koszul resultant matrices}
\cite{mantzaflaris2017resultants,buse_matrix_2017}. For both formulas,
the elements of the matrices are linear polynomials in the
coefficients of $(f_0,\dots,f_n)$. Other important resultant matrices
include \emph{B\'{e}zout-} and \emph{Dixon-type}; we refer
to \cite{emiris1999} and references therein for details.
We consider \emph{Koszul-type determinantal formulas} for
mixed multihomogeneous bilinear systems with two supports.

A well-known tool to derive determinantal formulas
\cite{weyman1994multigraded,dickenstein2003multihomogeneous,EmiMan-mhomo-jsc-12,emiris2016bit,mantzaflaris2017resultants,buse_matrix_2017}
is the Weyman complex \cite{weyman_calculating_1994},
a generalization of the Koszul complex.
For an introduction we refer to \cite[Sec.~9.2]{weyman2003cohomology}
and \cite[Sec.~2.5.C, Sec.~3.4.E]{gelfand2008discriminants}.
We follow this approach.

For \emph{unmixed} multihomogeneous systems, that is systems where all the polynomial share
the same support, determinantal formulas are well studied, e.g.,
\cite{sturmfels1994multigraded,weyman1994multigraded,kapur1997extraneous,chtcherba2000conditions,dandrea2001explicit,weyman2003cohomology,dickenstein2003multihomogeneous}.
On the other hand, when we consider polynomials with different
supports, that is \emph{mixed systems}, little is known about
determinantal formulas; with the exception of scaled multihomogeneous
systems \cite{EmiMan-mhomo-jsc-12}, that is when the supports
are scaled copies of one of them, and the bivariate tensor-product
case \cite{mantzaflaris2017resultants,buse_matrix_2017}.

The resultant is also a tool to solve $0$-dimensional square
polynomial systems $(f_1,\dots,f_n)$. There are different variants,
for example by hiding a variable, or using the u-resultant; we refer
to \cite[Chp.~3]{cox2006using} for a general introduction.
When a \emph{Sylvester-type formula} is available, we can use the
corresponding resultant matrix to obtain the matrix of the multiplication map of
a polynomial $f_0$ in $\K[\bm{x}] / \langle f_1,\dots,f_n \rangle$.  Then,
we can solve the system
by computing the eigenvalues and
eigenvectors of the latter  matrix, e.g.,
\cite{auzinger1988elimination,emiris1996complexity}.
The eigenvalues correspond to the evaluation of $f_0$ at every zero of
the system. From the eigenvectors we can recover the coordinates of
the zeros.
%
To our knowledge similar techniques involving matrices coming from
Koszul-type formulas do not exist up to now.

We consider mixed bilinear polynomial systems.  On the one hand, this
is simplest case of mixed multihomogeneous systems where no resultant
formula was known.  On the other hand, bilinear, and their
generalization multilinear, polynomial systems are common in
applications, for example in cryptography
\cite{FauLevPer-crypto-08,Joux-index-14} and game theory
\cite{EmiVid14}.
We refer to \cite{Faugere2011}, see also \cite{spaenlehauer-phd-12},
for computing the roots of unmixed multilinear systems by means of
Gr\"obner bases, and to \cite{emiris2016bit} by using resultants. We
refer to \cite{bender_mixedGB_2018} for a Gr\"obner bases approach to
solve square mixed multihomogeneous systems.

\paragraph*{\textbf{Our contribution}}
We introduce a new algorithm to solve square mixed multihomogeneous
systems consisting of bilinear polynomials with two different
supports. It relies on eigenvalues and eigenvectors computations.
Following classic resultant techniques we add a polynomial, $f_0$,  to make the system overdetermined.
The polynomial $f_0$ must be trilinear, as this is simplest one that can
separate the roots.
Then, we introduce a determinantal formula for the resultant of this
overdetermined system.  This is the first determinantal formula for a
mixed multilinear polynomial system.
Using Weyman's complex, we derive a \emph{Koszul-type} formula and
compute the resultant as the determinant of a \emph{Koszul resultant
  matrix}.

We present a general extension of the eigenvalue criterion that works
for a general class of formulas (see \Cref{def:piTheta}), which
include the Koszul-type and Sylvester-type formulas as special cases.
We consider a square matrix $M$ whose determinant is a multiple of the
resultant of a system $(f_0,\dots,f_n)$. If there is a monomial
$\bm{x}^\sigma$ in $f_0$ such that we can partition $M$ as
$\bigl[\begin{smallmatrix}M_{1,1} & M_{1,2} \\ M_{2,1} &
  M_{2,2} \end{smallmatrix} \bigr]$ where $M_{1,1}$ is invertible, the
coefficient of the monomial $\bm{x}^\sigma$ in $f_0$ appears solely in
the diagonal of $M_{2,2}$ and this diagonal contains only this
coefficient, then the evaluations of
$\frac{f_0(\bm{x})}{\bm{x}^\sigma}$ at the solutions of
$(f_1,\dots,f_n)$, that is
$\{\frac{f_0(\bm{x})}{\bm{x}^\sigma}|_{\bm{x} = \alpha} : (\forall i >
0) f_i(\alpha) = 0, \bm{x}^\sigma|_{\bm{x} = \alpha} \neq 0 \}$, are
eigenvalues of the Schur complement of $M_{2,2}$, that is
$M_{2,2} - M_{2,1} \cdot M_{1,1}^{-1} \cdot M_{1,2}$.

We extend the eigenvector criteria for these mixed bilinear
systems. When $M$ is our \emph{Koszul resultant matrix}, we show how
to recover the coordinates of the solutions from the eigenvectors of
the Schur complement of $M_{2,2}$. This approach works for systems
whose solutions have no multiplicities.  

Algorithm~\ref{alg:solve2bilinear} summarizes our strategy to solve
square 0-dimen\-sional 2-bilinear systems whose solutions have no
multiplicities.
    \begin{algorithm}
      \caption{$\texttt{Solve2Bilinear}((\bar{f}_1,\dots,\bar{f}_{n}))$}
      \begin{algorithmic}[1]
        \label{alg:solve2bilinear}
        \REQUIRE \parbox[t]{180px}{$(\bar{f}_1,\dots,\bar{f}_k)$ is a square 2-bilinear
        system such that $V_\P(\bar{f}_1,\dots,\bar{f}_k)$ is
        finite and has no multiplicities.}

        \STATE{$A \leftarrow $ Random linear change of coordinates
        preserving the structure.}

        \STATE{$(f_1,\dots,f_n) \leftarrow (\bar{f}_1 \circ A, \dots, \bar{f}_n \circ A)$.} \hfill (\Cref{thm:chgOfCoord})
        
        \STATE{$f_0 \leftarrow$ Random trilinear polynomial in
        $S(1,1,1)$.}

        \STATE{
        $ \bigl[\begin{smallmatrix}M_{1,1} & M_{1,2} \\ M_{2,1} &
          M_{2,2} \end{smallmatrix} \bigr]$ $\leftarrow$ 
        $\left\{ \; \text{\parbox{280px}{Matrix
        corresponding to $\delta_1((f_0,\dots,f_n),\bm{m})$, split
        wrt the monomial $\mon{\theta}$. \hfill (\Cref{def:piTheta})}} \right. $ }

      \STATE{$\left\{\left(\frac{f_0}{\mon{\theta}}(\alpha), \bar{v}_\alpha\right)\right\}_{\alpha} \leftarrow$ 
        $\left\{
          \text{
        \parbox{255px}{Set of pairs
          Eigenvalue-Eigenvector of the Schur complement of $M_{2,2}$. \hfill (\Cref{thm:eigenvalues})}
        } 
      \right. $}

        \FORALL{$\left(\frac{f_0}{\mon{\theta}}(\alpha), \bar{v}_\alpha\right) \in \left\{\left(\frac{f_0}{\mon{\theta}}(\alpha), \bar{v}_\alpha\right)\right\}_{\alpha}$}

        \STATE{\parbox{350px}{Extract the coordinates $\alpha_x,\alpha_y$ from $\rho_\alpha(\bm{\widehat{\lambda}_\alpha})$ by recovering it \linebreak from $\bigl[\begin{smallmatrix} M_{1,1}^{-1} \cdot M_{2,1} \\
            I \end{smallmatrix}\bigr] \cdot \bar{v}$. \hfill (\Cref{thm:eigenvectors})}}

      \STATE{\parbox{350px}{Let $\alpha_z \in \Pr^{n_z}$ be the unique
          solution to the linear system given by
          $\{f_1(\alpha_x,\alpha_y,\bm{z}) =
          0,\dots,f_n(\alpha_x,\alpha_y,\bm{z}) =0\}$, over $\K[\bm{z}]$.}}

      \STATE{Recover the solution of the system
        {$(\bar{f}_1,\dots,\bar{f}_n)$}, as $A\big((\alpha_x,\alpha_y,\alpha_z)\big)$.}

        \ENDFOR
         \end{algorithmic}
       \end{algorithm}

\paragraph*{Future work.}
Weyman complex leads to determinantal formulas for
mixed multihomogeneous systems. A possible extension is
to classify all the possible determinantal formulas for mixed
multihomogeneous systems of this
construction, 
similarly to \cite{weyman1994multigraded}.
The structure of the Koszul resultant matrix could lead
to more efficient algorithms to perform linear algebra with these
matrices, and hence to solve faster, theoretically and practically,
square mixed multihomogeneous systems.
Finally, our eigenvector criterion should be extensible to any Koszul
resultant matrix. This approach might be adapted to recover the
coordinates of the solutions with multiplicities.

\paragraph*{Paper organization} In \Cref{sec:preliminaries} we
introduce notation and the resultant of mixed
multihomogeneous systems.  In \Cref{sec:determinantal-form},  we present
the Weyman complex in our setting and we
prove the existence of a Koszul-type formula.
Then, in
\Cref{sec:solving-2-bilinear}, we present algorithms for solving
2-bilinear systems; \Cref{sec:eigenvalues-criteria} extends the
eigenvalue criterion to a general class of matrices and
\Cref{sec:eigenvectors} studies the eigenvectors to recover the
coordinates of the solutions. Finally, in \Cref{sec:complexity}, we
compare the size of our matrix with the experimental size of the
matrices in Gr\"{o}bner basis computation.

\section{Preliminaries}
\label{sec:preliminaries}
  Consider $n_x,n_y,n_z \in \N$ and let $\P := \Pr^{n_x} \times
  \Pr^{n_y} \times \Pr^{n_z}$ be a multiprojective space over an
  algebraic closed field $\K$ of characteristic $0$.
  Consider $\bm{x} := \{x_{0},\dots,x_{n_x}\}$,
  $\bm{y} := \{y_{0},\dots,y_{n_y}\}$,
  $\bm{z} := \{z_{0},\dots,z_{n_z}\}$ and let
  $S_x(d_x) := \K[\bm{x}]_{d_x}$, $S_y(d_y) := \K[\bm{y}]_{d_y}$, and
  $S_z(d_z) := \K[\bm{z}]_{d_z}$ be the spaces of homogeneous
  polynomials in variables $\bm{x}$, $\bm{y}$ and $\bm{z}$ and degrees
  $d_x$, $d_y$ and $d_z$, respectively. Let
  $S(d_x,d_y,d_z) := S_x(d_x) \otimes S_y(d_y) \otimes S_z(d_z)$ be
  the multihomogeneous polynomials in $\bm{x}$, $\bm{y}$, and $\bm{z}$
  of degrees $d_x$, $d_y$, and $d_z$, respectively. We say that the
  polynomials in $S(d_x,d_y,d_z)$ have multidegree
  $\bm{d} := (d_x,d_y,d_z) \in \Nz^3$.
  To avoid the repetition of the various definitions for $x$, $y$, and $z$,
  we consider $t \in \{x,y,z\}$.
  The dual space of $S_t(d_t)$ is $S_t(d_t)^*$.
  For 
  $\sigma_t \in \Nz^{n_t+1}$, we define
  $\bm{t}^{\sigma_t} := \prod_{i=0}^{n_t} t_{i}^{\sigma_{t,i}}$.
  Then $\A(d_t) := \{\sigma_t : \bm{t}^{\sigma_t} \in S_t(d_t)\}$ is
  the set of the exponents of all the monomials of degree $d_t$ in
  $\bm{t}$
  and $\A(\bm{d}) := \A(d_x) \times \A(d_y) \times \A(d_z)$ is  the
  set of all the exponents of the monomials of multidegree
  $\bm{d}$. If
  $\bm{\sigma} = (\sigma_x,\sigma_y,\sigma_z) \in \A(\bm{d})$,  then 
  $\mon{\sigma} :=
  \bm{x}^{\sigma_x}\bm{y}^{\sigma_y}\bm{z}^{\sigma_z}$.
  Let $n := n_x+n_y+n_z$. For multidegrees
  $\bm{d} = (\bm{d}_0,\dots,\bm{d}_n) \in (\Nz^{3})^{n+1}$, we
  consider  square multihomogeneous polynomial system
  \begin{equation}
    \label{eq:def-f-1}
    \bm{f} := (f_1,\dots,f_n) \in  S(\bm{d}_1) \times \dots \times S(\bm{d}_n)  \enspace .
  \end{equation}
  Let $V_\P(\bm{f})$ be the set of solutions of $\bm{f}$ over $\P$.
  The multihomogeneous B\'ezout bound (\MHB) \cite{van1978varieties}
  bounds the number of isolated solutions of $\bm{f}$ over $\P$~\cite{bernshtein1975number,kushnirenko1976newton,khovanskii1978newton}.
  The bound is attained for any \emph{generic} square system
  $\bm f$. 
  It is the mixed volume of the polytopes
  $\A(\bm{d}_1), \dots, \A(\bm{d}_n)$~\cite[Chp. 7]{cox2006using} and
  appears as the coefficient of the monomial
  $\prod_{t \in \{x, y, z \}} X_t^{n_t}$ in
  $\prod_{j = 1}^n \sum_{t \in \{x, y, z \}} \bm{d}_{j,t} X_t$~\cite{morgan1987homotopy}.

  In the sequel we consider overdetermined systems which we
  construct by adding an $f_0 \in S(\bm{d}_0)$ to $\bm{f}$, that is,
  \begin{equation}
    \label{eq:def-f-0}
    \bm{f_0} := (f_0,f_1,\dots,f_n) \in S(\bm{d}_0) \times \dots \times S(\bm{d}_n) \enspace .
  \end{equation}
  Typically, we will consider $\bm d_0 =(1,1,1)$, as we would like $f_0$
  to be as simple as possible while still depending on all the
  variables.

  \subsection{Multihomogeneous sparse resultant}
  \label{sec:sparse-res}
  The multihomogeneous sparse resultant of $\bm{f_0}$ is a polynomial
  in the coefficients of the polynomials in $\bm{f_0}$, which vanishes
  if and only if the system has a solution over $\P$.
  Following \cite{cox2006using}, for fixed
  $\bm{d}_0\dots\bm{d}_n \in \Nz^{3}$, we
  introduce a set of variables
  $\bm{u}_i := \{u_{i,\bm{\sigma}}\}_{\bm{\sigma} \in \A(\bm{d}_i)}$,
  for $0 \leq i \leq n$, and
  $\bm{u} := \{\bm{u}_0, \dots, \bm{u}_n\}$. Given $P \in \K[\bm{u}]$,
  we let $P(\bm{f_0})$ denote the value obtained by replacing each
  variable $u_{i,\bm{\sigma}}$ with the coefficient of the monomial
  $\mon{\sigma}$ in the polynomial $f_i$ of $\bm{f_0}$.
  In this way we obtain polynomials over the coefficients of a
  polynomial system.
  The ``universal'' system
  {$\bm{F_{\bm{d}_0,\dots,\bm{d}_n}}  \! \in
  \K[\bm{u_0}][\bm{x},\bm{y},\bm{z}] \times \dots \times
  \K[\bm{u_n}][\bm{x},\bm{y},\bm{z}]$} is
  \begin{align}
    \bm{F_{\bm{d}_0,\dots,\bm{d}_n}} :=
    \Big(\sum_{\bm{\sigma} \in \A(\bm{d}_0)} u_{0,\bm{\sigma}}
    \mon{\sigma}, \; \dots, \;
    \sum_{\bm{\sigma} \in \A(\bm{d}_n)} u_{n,\bm{\sigma}}
    \mon{\sigma} \Big) .
  \end{align}
Here the variables of $\bm{u}$ parametrize the systems described by
polynomials in $S(\bm{d}_0) \times \dots \times S(\bm{d}_n)$ over
$\K^{\#\A(\bm{d}_0)} \times \dots \times \K^{\#\A(\bm{d}_n)}$.

  Consider the set of all tuples of $n+1$ multihomogeneous polynomials
  together with their common solutions over $\P$,  $\{ (f_0,\dots,f_n,
  \alpha) \in S(\bm{d}_0) \times \cdots \times S(\bm{d}_n) \times \P :
  (\forall 0 \leq i \leq n) \, f_i(\alpha) = 0\}$. The projection of
  this set on $S(\bm{d}_0) \times \cdots \times S(\bm{d}_n)$ is
  the set of overdetermined systems with common solutions in $\P$, $\{
  (f_0,\dots,f_n) \in S(\bm{d}_0) \times \cdots \times S(\bm{d}_n) :
  V_\P(f_0,\dots,f_n) \neq \emptyset \}$. By the Projective Extension
  Theorem \cite[Chp. 8 Sec. 5]{cox1992ideals}, this projection is a
  closed set under the Zariski topology and it forms an irreducible
  hypersurface over the vector space $S(\bm{d}_0) \times \cdots \times
  S(\bm{d}_n)$ \cite[Chp.~8]{gelfand2008discriminants}. More formally,
  there is an irreducible polynomial $\Res_\P(\bm{d}_0,\dots,\bm{d}_n) \in \Z[\bm{u}]$
  such that for all the systems $\bm{f_0} \in S(\bm{d}_0) \times \cdots
  \times S(\bm{d}_n)$, $V_\P(\bm{f_0}) \neq \emptyset$ if and only if
  $\Res_\P(\bm{d}_0,\dots,\bm{d}_n)(\bm{f_0}) = 0$. This polynomial is the sparse
  resultant over $\P$ for multihomogeneous systems of
  multidegrees $(\bm{d}_0,\dots,\bm{d}_n)$.

  The resultant $\Res_\P(\bm{d}_0,\dots,\bm{d}_n)$ is itself a multihomogeneous polynomial,
  homogeneous in  each block of variables $\bm{u}_i$. For each $i$, its
  degree with respect to $\bm{u}_i$ is 
  $\MHB(\bm{d}_0,\dots,\bm{d}_{i-1},\bm{d}_{i+1},\dots,\bm{d_{n}})$.

  \subsection{2-bilinear systems}
  \label{sec:2-bilinear}
  A square \emph{2-bilinear system} of type $(n_x,n_y,n_z\;;\;r, s)$ is a
  bilinear system $\bm{f} := (f_1,\dots,f_n)$ with two different
  supports, namely $f_1,\dots,f_r \in S(1,1,0)$ and
  $f_{r+1},\dots,f_n \in S(1,0,1)$, such that $n = r + s$,
  $n_y \leq r$ and $n_z \leq s$.
  It holds $\MHB(\bm{f}) = {r \choose n_y}{n-r \choose n_z}$.
  
  \begin{example}
    The following (\Cref{eq:sysEx}) is a square 2-bilinear system of type $(1,1,1 \;;\;
    2,1)$ and has two solutions over $\P$, namely $\alpha_1 := (1\!:\!1 \;;\;
    1\!:\!1 \;;\; 1\!:\!1)$ and $\alpha_2 := (1\!:\!3 \;;\; 1\!:\!2
    \;;\; 1\!:\!3)$.
    \begin{align} \label{eq:sysEx}
      \left\{
       \begin{array}{ll}
         f_1 :=     7  x_0  y_0  -  8  x_0  y_1  -   x_1  y_0  +  2  x_1  y_1 \\
         f_2 :=  -  5  x_0  y_0  +  7  x_0  y_1  -   x_1  y_0  -     x_1  y_1 \\
         f_3 :=  -  6  x_0  z_0  +  9  x_0  z_1  -   x_1  z_0  -  2  x_1  z_1
       \end{array}
       \right. .
    \end{align}
  \end{example}

  Consider the trilinear $f_0 \in S(1,1,1)$. We refer to the systems
  $\bm{f_0} := (f_0,f_1,\dots,f_n)$ as overdetermined 2-bilinear
  systems.
  We can also consider $f_0$ in $S(1,1,0)$, $S(1,0,1)$,
  $S(1,0,0)$, $S(0,1,0)$ and $S(0,0,1)$. We work with a trilinear
  $f_0$ because in the other cases it is not always possible to
  separate all the solutions of $V_\P(\bm{f})$.
  
  \begin{example}[Cont.]
    Consider the overdetermined 2-bilinear system $\bm{f_0} :=
    (f_0,f_1,f_2,f_3)$, where
    \begin{align*} 
      f_0 := & \; 3 \, x_{0} y_{0} z_{0} - x_{0} y_{0} z_{1}
      - 4 \, x_{0} y_{1} z_{0} + 2 \, x_{0} y_{1} z_{1} \\ &
      + x_{1} y_{0} z_{0} + 2 \, x_{1} y_{0} z_{1}
      + 2 \, x_{1} y_{1} z_{0} - 2 \, x_{1} y_{1} z_{1} .
    \end{align*} 
  \end{example}

  In the following, we use $\bm{F^{(2)}}$ to denote the ``universal''
  system of overdetermined 2-bilinear systems
  (see~\Cref{sec:sparse-res}). Similarly, we use $\Res_\P^{(2)}$, for
  the resultant of the ``universal'' system $\bm{F^{(2)}}$.

  \begin{lemma}
    Let $\MHB(\bm{f}) = {r \choose n_y}{s \choose n_z}$. The degree of $\Res_\P^{(2)}$ is
    \begin{align} \label{eq:degRes} \mu := (n_x+1) \, \MHB(\bm{f}) \,
      \frac{r\cdot s - n_y\cdot n_z + r + s + 1} {(r-n_y+1)(s-n_z+1)}
      \enspace .
    \end{align}
  \end{lemma}
  
\section{Determinantal formulas for 2-bilinear systems}
\label{sec:determinantal-form}

  
A complex $K_{\bullet}$ is a sequence of modules $\{K_v\}_{v \in \Z}$
together with homomorphisms $\delta_v : K_v \rightarrow K_{v-1}$, such
that
$(\forall v \in \Z) \; \Im(\delta_v) \subseteq \Ker(\delta_{v-1})$,
i.e., $\delta_v \circ \delta_{v-1} = 0$. We say that the complex is
exact if $(\forall v \in \Z) \; \Im(\delta_v) = \Ker(\delta_{v-1})$. A
complex is bounded when there are two constants $a$ and $b$ such that
for every $v < a$ or $b < v$, it holds $K_v = 0$. If all the $K_v$ are
finite dimensional free-modules, then we can choose a basis of them
and we can represent the maps $\delta_v$ using matrices. 
Under certain
assumptions (see \cite[App.~A]{gelfand2008discriminants}) given a
bounded complex of finite dimensional free-modules we can define its
determinant. It is the quotient of minors of the matrices of
$\delta_v$ and it is not zero if and only if the complex is exact.  If
there are only two non-zero modules of the same dimensions in the
complex (that is all the other modules are the zero module), the
determinant of the complex reduces to the determinant of the (matrix
of the) map between these modules.

  The Weyman Complex
  \cite{weyman_calculating_1994,weyman1994multigraded,weyman2003cohomology}
  of a multihomogeneous system $\bm{f}$ is a bounded complex that is
  exact if and only if the sparse resultant of the system $\bm{f}$
  does not vanish \cite[Thm.~9.1.2]{weyman2003cohomology}. The
  determinant of the complex is a power of the resultant
  \cite[Prop.~9.1.3]{weyman2003cohomology}. When all the multidegrees
  are bigger than zero, the determinant of this complex is a non-zero
  constant multiple of the sparse resultant
  \cite[Thm.~3.4.11]{gelfand2008discriminants}. If the Weyman Complex
  only involves two non-zero modules, the resultant of the
  corresponding system is the determinant of the map between these
  modules, and it has a determinantal formula.

  Let 
  $\bm{f_0} := (f_0,f_1,\dots,f_n)$
  be an overdetermined 2-bilinear system.
%
  Consider $E := \K^{n+1}$ and its canonical basis
  $e_0,\dots,e_n$. Given a set $I \subset \{0,\dots,n\}$, we define
  $\bm{e}_{I} := e_{I_1} \wedge \dots \wedge e_{I_{\# I}}$ as the
  exterior product of the elements $e_{I_1},\dots,e_{I_{\# I}}$. As
  the exterior product is antisymmetric, that is
  $e_i \wedge e_j = - e_j \wedge e_i$, when we write
  $e_{I_1} \wedge \dots \wedge e_{I_{\# I}}$ we assume that
  $(\forall i) \, I_i < I_{i+1}$. Let $\bigwedge\limits_{a,b,c} E$ be the
  vector space over $\K$ generated by
  $\{ \bm{e}_{K \cup I \cup J} : K \subset \{0\}, I \subset
  \{1,\dots,r\}, J \subset \{r+1,\dots,n\}, \#I = a, \#J = b, \#K = c
  \}$.
  
  For a \textit{degree vector} $\bm{m} \in \Z^3$,
  the Weyman complex is $K_\bullet(\bm{f_0},\bm{m})$.
%
%
  Each module of the complex is $K_v(\bm{m}) :=
  \bigoplus_{p=0}^{n+1} K_{v,p}(\bm{m})$, where 
%
%
%
 $$ 
K_{v,p}(\bm{m}) := \bigoplus_{\substack{a+b+c = p \\ 0 \leq a \leq r \\ 0 \leq b \leq s \\ 0 \leq c \leq 1}}
 H^{p-v}_{\P}(\bm{m} - (p,p-b,p-a)) \otimes \bigwedge_{a,b,c} E ,$$
 and $H^{q}_{\P}(\bm{m'})$ is the $q$-th cohomology of $\P$ with
 coefficients in the sheaf $\mathcal O(\bm{m'})$, and the space of global
 sections is $H^0_{\P}(\bm{m'})$~\cite{Hart77}.
 Note that the terms $K_{v,p}(\bm{m})$ do not depend on $\bm{f_0}$~\cite[Prop.~2.1]{weyman1994multigraded}.
 Since $\P$ is a product of projective spaces, by K\"unneth's formula
  %
  \begin{align} \label{eq:kunneth}
    H^{p-v}_\P\left(m'_x,m'_y,m'_z\right) \cong \bigotimes_{t \in \{x,y,z\}} H^{j_t}_{\Pr^{n_t}}(m'_t) ,
  \end{align}
  where $j_x + j_y + j_z = p-v$.
  By Serre's duality~\cite[Ch.III,Thm.~5.1]{Hart77} we have the identifications:
  \begin{proposition}
    \label{th:serre}
    For each $t \in \{x,y,z\}$, $m'_t \in \Z$, it holds \linebreak
    (1)~$H^0_{\Pr^{n_t}}(m'_t) \cong S_t(m'_t)$ if $m'_t\ge 0$, 
    (2)~$H^{n_t}_{\Pr^{n_t}}(m'_t) \cong S_t(-m'_t-1-n_t)^*$ if $m'_t<n_t$, 
    where ``$*$'' denotes the dual space, and
    (3)~$H^q_{\Pr^{n_t}}(m'_t) \cong 0$, of all other values of $q$ and $m_t$.
  \end{proposition}
  As a corollary from \Cref{eq:kunneth}, for each
  $t \in \{x,y,z\}$, $j_t \in \{0, n_t\}$.
  Moreover, we can identify dual complexes.
  \begin{proposition}[{\cite[Thm.~5.1.4]{weyman2003cohomology}}] \label{prop:dualDegreeVectors} Let $\bm{m}$ and
    $\bm{m'}$ be degree vectors such that
    $\bm{m} + \bm{m'} = (n_y+n_z, n_x + n_z - s, n_x+n_y-r)$.  Then,
    {$K_v(\bm{m}) \cong K_{1-v}(\bm{m'})^*$} for all $v\in\ZZ$ and
    {$K_\bullet(\bm{f_0},\bm{m})$} is dual to
    {$K_\bullet(\bm{f_0},\bm{m'})$}.
\end{proposition}
    
  \subsection{Degree vectors and determinantal formulas}
  \label{sec:degree-vectors}
  
  If $K_1(\bm{m})$, $K_0(\bm{m})$ are the only non-zero modules in the
  Weyman complex $K_\bullet(\bm{f_0},\bm{m})$, then the determinant of
  the complex is the determinant of the map,
  between them, $\delta_1(\bm{f_0},\bm{m})$. In this case, we have a
  determinantal formula for the resultant. In the following, when it
  is clear from the context, we write $\delta_1$ instead of
  $\delta_1(\bm{f_0},\bm{m})$.
  
  \begin{theorem} \label{thm:goodDegreeVectors}
    Let $\bm{f_0}$ be a 2-bilinear overdetermined system of type \linebreak
    $(n_x,n_y,n_z;r,s)$, with $f_0\in S(1,1,1)$.
    The \textit{degree vectors} \linebreak
    (1)~$(n_y-1,-1,n_x+n_y-r+1)$,
    (2)~$(n_z+1,n_x+n_z-s+1,-1)$, \linebreak
    (3)~$(n_z-1,n_x+n_z-s+1,-1)$, 
    (4)~$(n_y+1,-1,n_x+n_y-r+1)$ \linebreak
      lead to determinantal Weyman complexes for
      $\Res_\P^{(2)}(\bm{f_0})$.
  \end{theorem}

\begin{observation} \label{obs:similarDegreeVectors}
  The four degree vectors of \Cref{thm:goodDegreeVectors} provide a
  single matrix formula.  Vector 1 (resp.~2) is obtained from 3
  (resp.~4) by exchanging the variables $\bm{y}$ and $\bm{z}$. By
  Prop.~\ref{prop:dualDegreeVectors}, we can see that 1,2 and
  3, 4 are dual pairs, yielding the same matrix transposed.
\end{observation}

  \begin{proof}
    We consider only the first degree vector
    $\bm{m} := \linebreak  (n_y-1,-1,n_x+n_y-r+1)$. By
    Obs.\ref{obs:similarDegreeVectors}, the other cases are similar.

    First, we show that the complex has only two non-zero terms.
    Since $K_v(\bm{m}) := \bigoplus_{p=0}^{n+1}
    K_{v,p}(\bm{m})$, and in view of 
%
    \Cref{eq:kunneth}, for each $K_{v,p}(\bm{m})$, we have
    to consider sums $\sum_{t \in \{x,y,z\}} j_t = p - v$. By
    \Cref{th:serre}, if $j_t \not\in \{0,n_t\}$, then $K_{v,p} =
    0$. The remaining cases are summarized in the following
    table and their analysis follows.

    {\scriptsize
      \begin{table}[!htbp]\centering
        \begin{tabular}{ | c | c | c || c | c | c | c | }
          \hline
          $j_x$ & $j_y$ & $j_z$ & \multicolumn{4}{c|}{Case}  \\
          \hline\hline
          $0$   & $\bm{0}$   & $0$   &  \multicolumn{4}{c|}{(1)}  \\ \hline
          $n_x$ & $\bm{0}$   & $0$   &  \multicolumn{4}{c|}{(1)} \\ \hline
        \end{tabular} \hfil
        \begin{tabular}{ | c | c | c || c | c | c | c | }
          \hline
          $j_x$ & $j_y$ & $j_z$ & \multicolumn{4}{c|}{Case}  \\
          \hline\hline
          $0$   & $\bm{0}$   & $n_z$ &  \multicolumn{4}{c|}{(1)} \\ \hline
          $n_x$ & $\bm{0}$   & $n_z$ &  \multicolumn{4}{c|}{(1)}  \\ \hline
        \end{tabular} \hfil
        \begin{tabular}{ | c | c | c || c | c | c | c | }
          \hline
          $j_x$ & $j_y$ & $j_z$ & \multicolumn{4}{c|}{Case}  \\
          \hline\hline
          $0$   & $n_y$ & $\bm{n_z}$ &  \multicolumn{4}{c|}{(2)} \\ \hline
          $n_x$ & $n_y$ & $\bm{n_z}$ &  \multicolumn{4}{c|}{(2)}  \\ \hline
        \end{tabular}\hfil
        \begin{tabular}{ | c | c | c || c | c | c | c | }
          \hline
          $j_x$ & $j_y$ & $j_z$ & \multicolumn{4}{c|}{Case}  \\
          \hline\hline
          $\bm{0}$   & $\bm{n_y}$ & $0$   &  \multicolumn{4}{c|}{(3)}   \\ \hline
          $\bm{n_x}$ & $\bm{n_y}$ & $\bm{0}$   &   \multicolumn{4}{c|}{(4)}  \\ \hline
        \end{tabular}
      \end{table}
    } 

    \noindent
    \textbf{Case 1: $\bm{j_y = 0}$.} 
    The second term in the tensor product of $K_{v,p}$ is
    $H^0_{\Pr^{n_y}}(-1 - a - c) \cong S_y(-1-a-c)$, by
    \Cref{th:serre}. As $a, c \geq 0$,
    $S_y(-1-a-c) = 0$. Hence,
    $K_{v,p} = 0$.

    \noindent
    \textbf{Case 2: $\bm{j_z = n_z}$.}
    The third term in the tensor product of $K_{v,p}$ is
    $H^{n_z}_{\Pr^{n_z}}(n_x + n_y - r + 1 - b - c) \cong
    S_z(-(n_x+n_y+n_z) + r -2 + b + c)^*$, by
    \Cref{th:serre}. As $n_x+n_y+n_z = r + s$,
    $- (n_x+n_y+n_z) + r -2 + b + c = - s - 2 + b + c < 0$ because
    $b \leq s$ and $c \leq 1$. Hence,
    $H^{n_z}_{\Pr^{n_z}}(n_x + n_y - r + 1 - b - c) = 0$ and so
    $K_{v,p} = 0$.
    
    \noindent
    \textbf{Case 3: $\bm{j_x = 0}$, $\bm{j_y = n_y}$.} 
    As $j_y = n_y$, the second term in the tensor product $K_{v,p}$ is
    $H^{n_y}_{\Pr^{n_y}}(-1 - a - c) \cong S_y(a+c-n_y)^*$, by
    \Cref{th:serre}. This module is not zero iff $a+c \geq
    n_y$. Consider the first term in the tensor product,
    $H^{0}_{\Pr^{n_x}}(n_y - 1 - p) \cong S_x(n_y -1 -p)$. If
    $a+c \geq n_y$, as $p = a+b+c$, then $n_y - 1 - p \leq -1 -b <
    0$. Hence, either $H^{n_y}_{\Pr^{n_y}}(-1 - a - c) = 0$ or
    $H^{0}_{\Pr^{n_x}}(n_y - 1 - p) = 0$, and so $K_{v,p} = 0$.

    \noindent
    \textbf{Case 4: $\bm{j_x = 0}$, $\bm{j_y = n_y}$, $\bm{j_z = 0}$.}
    The first term in the tensor product $K_{v,p}$ is
    $H^{n_x}_{\Pr^{n_x}}(n_y - 1 - p) \cong S_x(-n_x - n_y + p) =
    S_x(v)$, as $p - v = j_x + j_y + j_z = n_x + n_y$. Hence,
    $H^{n_x}_{\Pr^{n_x}}(n_y - 1 - p) \neq 0$ iff $v \geq 0$.
    As $j_z = 0$ the third term in the tensor product of $K_{v,p}$ is
    $H^{0}_{\Pr^{n_z}}(n_x + n_y - r + 1 - b - c) \cong S_z(n_x + n_y
    - r + 1 - b - c)$. This term is not zero iff
    $n_x + n_y - r + 1 \geq b + c$.
    Moreover, as $p = a + b + c$, $v = a + b + c - n_x - n_y$. Then,
    if $H^{0}_{\Pr^{n_z}}(n_x + n_y - r + 1 - b - c) \neq 0$, then
    $v \leq a - r + 1$. By definition $a \leq r$, so $v \leq 1$.

    We deduce that all other modules apart from
    $K_{1, n_x+n_y+1}(\bm{m})$ and \linebreak $K_{0, n_x+n_y}(\bm{m})$ are equal
    to zero. Hence, by~\cite[Prop.~9.1.3]{weyman2003cohomology} the
    determinant of (a matrix expressing) $\delta_1$ is a
    power
    \footnote{The exponent is known to be one for any very ample
      supports~\cite{gelfand2008discriminants}, i.e. 
      $(\forall i,j)\; \bm d_{i,j}>0$. However, due to
      the zero degrees, 2-bilinear supports are ample but not very ample.}
    of $\Res_\P^{(2)}(\bm{f_0})$.

    To conclude, it suffices to show that the exponent is
    equal to one.  Due to the form
    $\delta_1 : K_{1,q+1}(\bm{m}) \to K_{0,q}(\bm{m})$, the elements
    in a matrix that represents $\delta_1$ have degree $(q+1)-q=1$
    as polynomials in
    $\K[\bm{u}]$~\cite[Prop.~5.2.4]{weyman2003cohomology}.  Therefore,
    the exponent is one iff the degree of the resultant is
    equal to the dimension of the matrix of

    $$K_\bullet(\bm{f_0}, \bm{m}) : 0 \rightarrow K_{1,
      n_x+n_y+1}(\bm{m}) \xrightarrow{\delta_1} K_{0,
      n_x+n_y}(\bm{m}) \rightarrow 0 \, . $$
    We analyze the possible values for $(a,b,c)$ to compute the
    dimension. Following \textbf{Case 4}, if
    {$H^{0}_{\Pr^{n_z}}(n_x + n_y - r + 1 - b - c) \neq 0$},
    then
    the possible values for $a$ are $v + r - 1 \leq a \leq r$, for
    $v \in \{0,1\}$.
    As $b = p - a - c$, and $0 \leq c \leq 1$, we enumerate all the
    options for $(a,b,c)$ and write our modules as

    \begin{align} \label{eq:K1}
      K_{1} = & \; K_{1, n_x+n_y+1}  \cong \;  L_{1,1} \oplus L_{1,2} \\ = &
      \Big(
      S_x(1)^* {\otimes} 
      S_y(r-n_y)^* {\otimes} 
      S_z(0) 
      \otimes
      \bigwedge_{r,s-n_z+1,0} E
      \Big)
      \oplus \nonumber  \\ & 
      \Big(
      S_x(1)^* {\otimes} 
      S_y(r-n_y+1)^* {\otimes} 
      S_z(0) 
      \otimes
      \bigwedge_{r,s-n_z,1} E
      \Big) . \nonumber
      \\ \label{eq:K0}
      K_{0} = & \; K_{0, n_x+n_y} \cong \; L_{0,1} \oplus L_{0,2}  \oplus L_{0,3}  \oplus L_{0,4}  \\ 
      = & \Big(
      S_x(0)^* {\otimes} 
      S_y(r-n_y-1)^* {\otimes} 
      S_z(0) 
      \otimes
      \bigwedge_{r-1,s-n_z+1,0} E
      \Big)
      \oplus \nonumber \\ &
      \Big(
      S_x(0)^* {\otimes} 
      S_y(r-n_y)^* {\otimes} 
      S_z(1) 
      \otimes
      \bigwedge_{r,s-n_z,0} E
      \Big)
      \oplus \nonumber  \\ &
      \Big(
      S_x(0)^* {\otimes} 
      S_y(r-n_y)^* {\otimes} 
      S_z(0) 
      \otimes
      \bigwedge_{r-1,s-n_z,1} E
      \Big)
      \oplus  \nonumber  \\  &
      \Big(
      S_x(0)^* {\otimes} 
      S_y(r-n_y+1)^* {\otimes} 
      S_z(1) 
      \otimes
      \bigwedge_{r,s-n_z-1,1} E
      \Big) . \nonumber
    \end{align}
  
     To compute their dimensions we
    notice that
    {$\dim\left( \bigwedge_{a,b,c} E \right) = {r \choose a} {s \choose
      b} $ }, and we recall that $\dim S_t(q) = \dim S_t(q)^*= {n_t+q \choose q}$.
    %
    %
%
    The calculation leads to $\dim(K_{1}) = \dim(K_{0}) = \mu$,
    see \Cref{eq:degRes}.
  \end{proof}

  The four degree vectors of \Cref{thm:goodDegreeVectors} are not the
  only ones that lead to determinantal formulas. We are interested in
  them because, experimentally, there are no Sylvester-type formulas
  and only these degree vectors lead to Koszul-type formulas
  \cite{emiris2016bit,mantzaflaris2017resultants}.
  
  \subsection{Construction of the map $\bm{\delta_1(f_0,m)}$}
  \label{sec:determinantalMap}

  Following \cite[Sec.~5.5]{weyman2003cohomology}, we construct the
  map
  $\delta_1(\bm{f_0},\bm{m}) : K_{1}(\bm{m}) \rightarrow
  K_{0}(\bm{m})$. By \Cref{obs:similarDegreeVectors}, we only consider
  $\bm{m} = (n_y-1,-1,n_x+n_y-r+1)$.

  In the proof of \Cref{thm:goodDegreeVectors} we saw that the map
  $\delta_1(\bm{F^{(2)}},\bm{m})$ has linear coefficients in
  $\K[\bm{u}]$.  As it is a linear map between free modules, it is enough
  to define it over a basis of $K_0$ and $K_1$.

  First we introduce some notation.
  Let $t \in \{x,y,z\}$. For each ${\sigma_t} \in \A(d)$, $d\in\N_0$, consider
  $\bm{\partial t}^{\sigma_t} \in S_t(d)^*$ such that 
  $\bm{\partial t}^{\sigma_t}(\sum c_{\theta_t} \bm{t}^{\theta_t}) =
  c_{\sigma_t}$. The set
  $\{\bm{\partial t}^{\sigma_t} : {\sigma_t} \in \A(d)\}$ forms a
  basis of $S_t(d)^*$.
  The map
  $\star_t : \K[\bm{t}] \times \K[\bm{t}]^* \rightarrow \K[\bm{t}]^*$,
   acts as 
  $(\bm{t}^{\theta_t}, \bm{\partial t}^{\sigma_t}) \mapsto
  \bm{t}^{\theta_t} \star_t \bm{\partial t}^{\sigma_t}$, where
    \begin{align} \label{eq:star}
              \bm{t}^{\theta_t} \star_t \bm{\partial t}^{\sigma_t} = \left\{
                \begin{array}{ll}
                  \bm{\partial t}^{{\sigma_t}-{\theta_t}} &
                  \mathrm{if} \; (\forall i, \, 0 \leq i \leq n_t) \;
                  \sigma_{t,i} \geq \theta_{t,i} \\
                  0 & \mathrm{otherwise}
                \end{array}
              \right. .
    \end{align}
    This map is graded, that is, for each $(d,\bar{d}) \in \Z^2$, it
    maps the elements in $S_t(d) \times S_t(\bar{d})^*$ to
    $S_t(\bar{d} - d)^*$. We will denote the map by ``$\star$'' when
    the variable is clear from the context.
    We define the graded map $\psi$,
      \begin{multline} \label{eq:defpsi}
        \psi : \left(\K[\bm{x}]^* \otimes \K[\bm{y}]^* \otimes
        \K[\bm{z}]\right) \times \left( \K[\bm{x}] \otimes \K[\bm{y}]
        \otimes \K[\bm{z}]\right) \\ 
        \rightarrow \left(\K[\bm{x}]^*
        \otimes \K[\bm{y}]^* \otimes \K[\bm{z}]\right)
    \end{multline}
    \begin{multline*}
    \psi(\bm{\partial x}^{\sigma_x} \otimes \bm{\partial y}^{\sigma_y}
    \otimes \bm{z}^{\sigma_z}, \bm{x}^{\theta_x} \otimes
    \bm{y}^{\theta_y} \otimes \bm{z}^{\theta_z}) := \\
    (\bm{x}^{\theta_x} \star
    \bm{\partial x}^{\sigma_x}) \otimes
    (\bm{y}^{\theta_y} \star
    \bm{\partial y}^{\sigma_y}) \otimes
    (\bm{z}^{\theta_z + \sigma_z})
    \end{multline*}
    For each
    $(d_x,d_y,d_z,\bar{d_x},\bar{d_y},\bar{d_z}) \in \Z^6$, it maps
    $\left(S_x(d_x)^* \otimes S_y(d_y)^* \otimes \right.$ $\left.
    S_z(d_z)\right) \times\left( S_x(\bar{d_x})^* \otimes
    S_y(\bar{d_y})^* \otimes S_z(\bar{d_z})\right)$ to
    $S_x(d_x-\bar{d_x})^* \otimes S_y(d_y-\bar{d_y})^* \otimes
    S_z(d_z+\bar{d_z})$.
    
    \vspace{7px}
    
    As $\delta_1(\bm{f_0},\bm{m}) : K_1 \rightarrow K_0$ is linear and
    $K_1 \cong L_{1,1} \oplus L_{1,2}$, we define the map over a
    basis of $L_{1,1}$ and $L_{1,2}$.
    For each
    $\bm \ell \in S_x(1)^* {\otimes} \linebreak S_y(r-n_y)^* {\otimes} S_z(0)$ and
    $\bm{e}_I \in \!\!\!\!\! \bigwedge\limits_{r,s-n_z+1,0}\!\!\!\!\! E$, we consider
    $ \bm \ell \otimes \bm{e}_{I} \in L_{1,1}$ and 
    {
    $$
    \delta_1(\bm{f_0},\bm{m})\left(
      \bm \ell \otimes \bm{e}_{I}
    \right) :=  \!\!\!\!
    \sum_{i = 1}^{n_x+n_y+1} \!\!\!\! 
    (-1)^{i-1} \psi\left(\bm \ell, f_{I_i}\right)
    \otimes \bm{e}_{I \setminus \{I_i\}} \, \in L_{0,1} \oplus L_{0,2} .
    $$}

  For each
  $\bm \ell \in S_x(1)^* {\otimes} S_y(r-n_y+1)^* {\otimes} S_z(0)$ and
  $\bm{e}_J \in \!\!\! \bigwedge\limits_{r,s-n_z,1} \!\!\! E$, we consider
  $ \bm \ell \otimes \bm{e}_{J} \in L_{1,2}$ and
    {
    $$
    \delta_1(\bm{f_0},\bm{m})\left(\bm\ell \otimes
    \bm{e}_J\right) := \!\!\!\!\!
  \sum_{i = 1}^{n_x+n_y+1} \!\!\!\!\!
  (-1)^{i-1} \psi(\bm\ell, f_{J_i})
    \otimes \bm{e}_{J \setminus \{J_i\}} \,\in L_{0,2} \oplus L_{0,3} \oplus L_{0,4} .
    $$
    }
    %

    The map $\delta_1(\bm{f_0},\bm{m})$ corresponds to a Koszul-type
    formula, involving multiplication and dual multiplication maps. The
    matrix that represents this map is a Koszul resultant matrix
    \cite{mantzaflaris2017resultants,buse_matrix_2017}.
    
  \begin{example}[Cont.]
      \label{ex:partMatrix}
      In this case, $\bm{m} = (0,-1,1)$. We consider the following monomial
      basis,

    \begin{center}
    {
    \begin{tabular}{|c|c|}
      \hline
      \multicolumn{2}{|c|}{Basis of $K_1$ (Columns)}  \\ \hline\hline
    (A) & $ {\partial x_0} {\partial y_1^2} {\bm{e}}_{\{0,1,2\}} $ \\ \hline
    (B) & $ \partial x_1 \partial y_0^2 {\bm{e}}_{\{0,1,2\}} $  \\ \hline
    (C) & $ \partial x_1 \partial y_1^2 {\bm{e}}_{\{0,1,2\}} $  \\ \hline
    (D) & $ \partial x_0 \partial y_0 {\bm{e}}_{\{1,2,3\}} $ \\ \hline
    (E) & $ \partial x_0 \partial y_1 {\bm{e}}_{\{1,2,3\}} $  \\ \hline
    (F) & $ \partial x_1 \partial y_0 {\bm{e}}_{\{1,2,3\}} $  \\ \hline
    (G) & $ \partial x_1 \partial y_1 {\bm{e}}_{\{1,2,3\}} $   \\ \hline
    (H) & $ \partial x_0 \partial y_0\partial y_1 {\bm{e}}_{\{0,1,2\}} $  \\ \hline
    (I) & $ \partial x_0 \partial y_0^2 {\bm{e}}_{\{0,1,2\}} $ \\ \hline
    (J) & $ \partial x_1 \partial y_0\partial y_1 {\bm{e}}_{\{0,1,2\}} $  \\ \hline
    \end{tabular}\quad\begin{tabular}{|c|c|}
      \hline
      \multicolumn{2}{|c|}{Basis of $K_0$ (Rows)} \\ \hline\hline
    (I) & $ {\bm{e}}_{\{1,3\}}$ \\ \hline
    (II) & $ {\bm{e}}_{\{2,3\}}$ \\ \hline
    (III) & $ \partial y_0 {\bm{e}}_{\{0,1\}}$  \\ \hline
    (IV) & $ \partial y_1 {\bm{e}}_{\{0,1\}}$  \\ \hline
    (V) & $ \partial y_0 {\bm{e}}_{\{0,2\}}$  \\ \hline
    (VI) & $ \partial y_1 {\bm{e}}_{\{0,2\}}$  \\ \hline
    (VII) & $ \partial y_0 z_1 {\bm{e}}_{\{1,2\}}$  \\ \hline
    (VIII) & $ \partial y_1  z_1 {\bm{e}}_{\{1,2\}}$  \\ \hline 
    (IX) & $ \partial y_0  z_0 {\bm{e}}_{\{1,2\}}$  \\ \hline
    (X) & $ \partial y_1  z_0 {\bm{e}}_{\{1,2\}}$  \\ \hline
    \end{tabular}
    }
    \end{center}
    The following matrix represents
    $\delta_1(\bm{f_0},\bm{m})$ wrt the basis above.
    \begin{center}
      {\small
      \[ 
        \begin{array}{c || cccccccc | cc}
          & 
          (A) & (B) & (C) & (D) & (E) & (F) & (G) & (H) & (I) & (J)
          \\ \hline \hline
          %
         (I)  & 0 & 0 & 0 & \bm{5} & \bm{-7} & \bm{1} & \bm{1} & 0  & 0 & 0\\
         (II)  &  0 & 0 & 0 & \bm{7} & \bm{-8} & \bm{-1} & \bm{2} & 0 & 0 & 0 \\
         (III)  &  0 & \bm{-1} & 0 & 0 & 0 & 0 & 0 & \bm{-1}  & \bm{-5} & \bm{7}\\
         (IV)  &  \bm{7} & 0 & \bm{-1} & 0 & 0 & 0  & 0 & \bm{-1} & 0 & \bm{-5}\\
         (V)  &  0 & \bm{1} & 0 & 0  & 0 & 0 & 0 & \bm{-2} & \bm{-7} & \bm{8}\\
         (VI)  & \bm{8} & 0  & \bm{-2} & 0 & 0 & 0 & 0 & \bm{1} & 0 & \bm{-7}\\
         (VII)  &  0 & \bm{2} & 0 & \bm{9} & 0 & \bm{-2} & 0 & \bm{-2} & \bm{-1} & \bm{2}\\
         (VIII)  &  \bm{2} & 0 & \bm{-2} & 0 & \bm{9} & 0 & \bm{-2} & \bm{2} & 0 & \bm{-1}\\ \hline
         (IX)  &  0 & \bm{1} & 0 & \bm{-6} & 0 & \bm{-1} & 0 & \bm{2} & \bm{3} & \bm{-4} \\
         (X)  &  \bm{-4} & 0 & \bm{2} & 0 & \bm{- 6} & 0 & \bm{-1} & \bm{1} & 0 & \bm{3} 
         \end{array}
       \]
    }
    \end{center}
  The $2\times 2$ splitting illustrated above will be used in the next section.
  \end{example}

\section{Solving 2-bilinear systems}
\label{sec:solving-2-bilinear}

Consider a 0-dimensional system $f_1,\dots,f_n \in \K[\bm{x}]$. A
common strategy for solving is to work over
$\K[\bm{x}]/\langle f_1,\dots,f_n \rangle$, which is a finite a
dimensional vector space over $\K$. We fix a monomial basis,
choose $f_0 \in \K[\bm{x}]$, and compute the matrix that
represents the multiplication by $f_0$ in the quotient ring.
Its eigenvalues are the evaluations of $f_0$ at the
solutions. For a suitable basis, from the eigenvectors we
can recover the coordinates of all the solutions
\cite{elkadi2007introduction,cox2006using,cox2005inDickensteinEmiris}.
To compute these matrices we can use the Sylvester-type formulas
\cite{auzinger1988elimination,emiris1996complexity,cox2006using}.
We extend these techniques to a general family of matrices, that
includes the  Koszul resultant matrix (\Cref{sec:determinantalMap}).

  \subsection{Eigenvalues criteria}
  \label{sec:eigenvalues-criteria}
  In this section we assume fixed  multidegrees
  $\bm{d}_0,\dots,\bm{d}_n$.

  \begin{definition}[property $\Pi_\theta$] \label{def:piTheta}
    Given $\theta \in \A(\bm{d}_0)$ and a matrix \linebreak
    $M := \bigl[\begin{smallmatrix}M_{1,1} & M_{1,2} \\ M_{2,1} &
      M_{2,2} \end{smallmatrix} \bigr] \in \K[\bm{u}]^{\mathcal{K}
      \times \mathcal{K}}$ (\Cref{sec:sparse-res}), we say
    that $M$ has the property $\Pi_\theta(\bm{d}_0,\dots,\bm{d}_n)$,
    or simply $\Pi_\theta$, when:
    \begin{itemize}   \setlength\itemsep{-2px}
    \item $Res_\P(\bm{d}_0,\dots,\bm{d}_n)$ divides $\det(M)$,
    \item the submatrix $M_{2,2}$ is square and its diagonal entries
      equal to $u_{0,\bm{\theta}}$, and
    \item the coefficient $u_{0,\bm{\theta}}$ does not appear anywhere
      in $M$ expect from the diagonal of $M_{2,2}$.
    \end{itemize}
  \end{definition}

  For a system $\bm{f_0}$, Eq.~\eqref{eq:def-f-0}, let $M(\bm{f_0})$
  be the specialization of $M$ at $\bm{f_0}$
  (see~\Cref{sec:sparse-res}). If $M_{1,1}(\bm{f_0})$ is invertible,
  then the Schur complement of $M_{2,2}(\bm{f_0})$ is
  $M_{2,2}(\bm{f_0}) - M_{2,1}(\bm{f_0}) \cdot
  (M_{1,1}(\bm{f_0}))^{-1} \cdot M_{1,2}(\bm{f_0})$.  To simplify, we
  write
  $(M_{2,2} - M_{2,1} \cdot M^{-1}_{1,1} \cdot M_{1,2})(\bm{f_0})$.
  
  \begin{theorem}
    \label{thm:eigenvalues}
    Consider 
    $\theta \in \A(\bm{d}_0)$ and a matrix
    $M \in \K[\bm{u}]^{\mathcal{K} \times \mathcal{K}}$ such that
    $\Pi_\theta$ holds (Def.~\ref{def:piTheta}). Assume a system
    $\bm{f_0}$, Eq.\eqref{eq:def-f-0}, such that the specialization
    $M_{1,1}(\bm{f_0})$ is non-singular. Then, for all
    $\alpha \in V_\P(\bm{f})$ such that $\mon{\theta}(\alpha) \neq 0$,
    $\frac{f_0}{\mon{\theta}}(\alpha)$ is an eigenvalue of the Schur
    complement of $M_{2,2}(\bm{f_0})$.
  \end{theorem}
  
  \begin{proof}
    The idea of the proof is as follows: For each
    $\alpha \in V_\P(\bm{f})$, Eq.~(\ref{eq:def-f-1}), we consider a
    system $\bm{g_0}$, slightly different from $\bm{f_0}$, with 
    $\alpha$ as a solution. We study the matrices $M(\bm{f_0})$ and
    $M(\bm{g_0})$ and from the kernel of $M(\bm{g_0})$ we construct an
    eigenvector for the Schur complement of $M_{2,2}(\bm{f_0})$
    corresponding to an eigenvalue equal to
    $\frac{f_0}{\mon{\theta}}(\alpha)$.

    Let $\alpha \in V_\P(\bm{f})$ such that
    $\mon{\theta}(\alpha) \neq 0$. Consider the polynomial
    $g_0 := f_0 - \frac{f_0}{\mon{\theta}}(\alpha) \cdot \mon{\theta}$
    and a new system $\bm{g_0} := (g_0,f_1, \dots,f_n)$.
    The coefficients of the polynomials $g_0$ and $f_0$ are the same,
    with exception of the coefficient of the monomial $\mon{\theta}$,
    so the specializations $u_{i,\bm{\sigma}}(\bm{f_0})$ and
    $u_{i,\bm{\sigma}}(\bm{g_0})$ (\Cref{sec:sparse-res}) differ
    if and only if $i = 0$ and $\bm{\sigma} = \bm{\theta}$.
    Hence, as $\Pi_\theta$ holds, $u_{0,\bm{\theta}}$ does not appear
    in $M_{1,1}$, $M_{2,1}$, and $M_{1,2}$, and
    $M_{1,1}(\bm{g_0}) = M_{1,1}(\bm{f_0})$,
    $M_{1,2}(\bm{g_0}) = M_{1,2}(\bm{f_0})$, and
    $M_{2,1}(\bm{g_0}) = M_{2,1}(\bm{f_0})$.
    The specialization of $u_{0,\bm{\theta}}$ is a ring homomorphism,
    so
    $u_{0,\bm{\theta}}(\bm{g_0}) = u_{0,\bm{\theta}}(\bm{f_0}) -
    \frac{f_0}{\mon{\theta}}(\alpha)$. By $\Pi_\theta$,
    $u_{0,\bm{\theta}}$ only appears in the diagonal of
    $M_{2,2}$. Hence,
    $M_{2,2}(\bm{g_0}) = M_{2,2}(\bm{f_0}) -
    \frac{f_0}{\mon{\theta}}(\alpha) \cdot I$, where $I$ is the
    identity matrix.
    Therefore, $$M(\bm{g_0}) = \left[\begin{smallmatrix}M_{1,1} &M_{1,2}
        \\ M_{2,1} & M_{2,2} \end{smallmatrix} \right](\bm{f_0}) -
    \frac{f_0}{\mon{\theta}}(\alpha) \cdot
    \Big[\begin{smallmatrix}0&0\\0&I\end{smallmatrix} \Big].$$

    By construction $g_0(\alpha) = 0$, $\alpha \in V_\P(\bm{f})$, thus
    $\alpha \in V_\P(\bm{g_0})$, and so $\Res_\P(\bm{g_0})$
    vanishes. By property $\Pi_\theta$, $\det(M)$ is a multiple of
    $\Res_\P(\bm{d}_0,\dots,\bm{d}_n)$, hence $M(\bm{g_0})$ is    singular.
    Let $v \in \ker(M(\bm{g_0}))$,~then
    \[
      M(\bm{g_0}) \cdot v = 0 \iff
      \left[\begin{smallmatrix}
        M_{1,1} & M_{1,2} \\
        M_{2,1} & M_{2,2}
      \end{smallmatrix} \right](\bm{f_0}) \cdot v
      = \frac{f_0}{\mon{\theta}}(\alpha) \cdot
      \left[\begin{smallmatrix}
        0 & 0 \\
        0 & I
      \end{smallmatrix} \right] \cdot v \enspace .
      \]
      Multiplying this equality by the non-singular matrix
      related to the Schur complement of $M_{2,2}(\bm{f_0})$,
      $\left[\begin{smallmatrix} I & 0 \\ - M_{2,1} \cdot M_{1,1}^{-1}
          & I \end{smallmatrix}\right](\bm{f_0})$, we obtain
      \[
      \left[\begin{smallmatrix}
        M_{1,1} & M_{1,2} \\
        0 & (M_{2,2} - M_{2,1} \cdot M^{-1}_{1,1} \cdot M_{1,2})
      \end{smallmatrix} \right](\bm{f_0}) \cdot v
      = \frac{f_0}{\mon{\theta}}(\alpha) \cdot
      \left[\begin{smallmatrix} 0 & 0 \\ 0 & I
      \end{smallmatrix} \right] \cdot v \enspace .
      \]
      Consider the lower part of the matrices in the previous
      identity,
      \[
      \left[\begin{array}{c|c} \!0 & \text{ $M_{2,2} - M_{2,1} \cdot
        M^{-1}_{1,1} \cdot M_{1,2} \!$}
      \end{array} \right](\bm{f_0}) \cdot v
      = \frac{f_0}{\mon{\theta}}(\alpha) \cdot \left[\begin{array}{c|c} \! 0 & I \!
      \end{array} \right] \cdot v
      \]
    and let $\bar{v} := \left[\begin{array}{c|c} 0 & I
      \end{array} \right] \cdot v$ be a truncation of the vector $v$. Then,
   \[(M_{2,2} - M_{2,1} \cdot M^{-1}_{1,1} \cdot M_{1,2})(\bm{f_0}) \cdot \bar{v} =
      \frac{f_0}{\mon{\theta}}(\alpha) \cdot \bar{v} \enspace .\]

    This equality proves that $\frac{f_0}{\mon{\theta}}(\alpha)$ is an
    eigenvalue of the Schur complement of $M_{2,2}(\bm{f_0})$ with
    eigenvector $\bar{v}$.
  \end{proof}

  Let $\bm{f} \in S(\bm{d}_1) \times \dots \times S(\bm{d}_n)$,
  Eq.~(\ref{eq:def-f-1}), be a square system. Consider
  $f_0 \in S(\bm{d}_0)$ and $\theta \in \A(\bm{d}_0)$. We say that the
  rational function $\frac{f_0}{\mon{\theta}}$ separates the zeros of
  the system, if for all $\alpha \in V_\P(\bm{f})$,
  $\mon{\theta}(\alpha) \neq 0$ and for all
  $\alpha, \alpha' \in V_\P(f_1,\dots,f_n)$,
  $\frac{f_0}{\mon{\theta}}(\alpha) = \frac{f_0}{\mon{\theta}}(\alpha')
  \iff \alpha = \alpha'$.
 
  \begin{corollary}
    \label{thm:separateRootsImpliesDiagonalizable}
    Under the assumptions of \Cref{thm:eigenvalues}, if the row
    dimension of $M_{2,2}$ is $\MHB(\bm{d}_1,\dots,\bm{d}_n)$,
    $\frac{f_0}{\mon{\theta}}$ separates the zeros of
    $(f_1,\dots,f_n)$ and there are $\MHB(\bm{d}_1,\dots,\bm{d}_n)$
    different solutions for this subsystem (over $\P$), then the Schur
    complement of $M_{2,2}(\bm{f_0})$ is diagonizable with
    eigenvalues $\frac{f_0}{\mon{\theta}}(\alpha)$,
    for $\alpha \in V_\P(f_1, \dots, f_n)$.
  \end{corollary}

  \begin{proof}
    As a consequence of \Cref{thm:eigenvalues}, for each
    $\alpha \in V_\P(\bm{f})$ we have an eigenvalue
    $\frac{f_0}{\mon{\theta}}(\alpha)$ for the Schur complement of
    $M_{2,2}(\bm{f_0})$. As $\frac{f_0}{\mon{\theta}}$ separates these
    zeros, all the eigenvalues are different. Hence, we have as many
    different eigenvalues as the dimension of the matrix, so the
    matrix is diagonalizable.
  \end{proof}

  Note that, as the $\MHB$ bounds the number of isolated solutions
  counting multiplicities, we can not use
  Thm.~\ref{thm:separateRootsImpliesDiagonalizable} when we have a square
  system $\bm{f}$ such that its solutions over $\P$ have
  multiplicities.
  
  \begin{lemma}
    \label{thm:allTheEigenvaluesAreRoots}
    Under the assumptions of \Cref{thm:eigenvalues}, assume that
     \linebreak
    $\Res_\P(\bm{f_0}) \neq 0$ and
    $\det(M) = q \cdot \Res_\P(\bm{d}_0,\dots,\bm{d}_n)$, where $q$ is
    a non-zero constant in $\K$. If $\lambda$ is an eigenvalue of the
    Schur complement of $M_{2,2}(\bm{f_0})$, then there is
    $\alpha \in V_\P(\bm{f})$ such that
    $\lambda = \frac{f_0}{\mon{\theta}}(\alpha)$.
  \end{lemma}

  \begin{proof}
    Consider the system
    $\bm{g_0} := ((f_0-\lambda \cdot \mon{\theta}),f_1,\dots,f_n)$. As
    the matrix of the Schur complement in the proof of
    \ref{thm:eigenvalues} is invertible, we extend $\bar{v}$ to
    $v = \bigl[\begin{smallmatrix} M_{1,1}^{-1} \cdot M_{2,1} \\
      I \end{smallmatrix}\bigr](\bm{f_0}) \, \bar{v}$, and reverse the
    argument in this proof to show that $M(\bm{g_0})$ is singular. As
    the determinant of $M$ is a non-zero constant multiple of the
    resultant, we deduce that $Res_\P(\bm{g_0})$ is zero. Let
    $\alpha \in V_\P(\bm{g_0})$, then $\alpha \subset V_\P(\bm{f})$
    and $(f_0-\lambda \cdot \mon{\theta})(\alpha) = 0$, equivalently,
    $f_0(\alpha) = \lambda \cdot \mon{\theta}(\alpha)$. As we assumed
    that $\Res_\P(\bm{f_0}) \neq 0$, then $f_0(\alpha) \neq 0$ and so
    $\frac{f_0}{\mon{\theta}}(\alpha) = \lambda$.
  \end{proof}

  \begin{proposition}
    \label{thm:m11invertible}
    Under the assumptions of \Cref{thm:eigenvalues}, assume
    $\det(M) = q \cdot \Res_\P(\bm{d}_0,\dots,\bm{d}_n)$, where $q$ is
    a non-zero constant in $\K$, and that the (row) dimension of
    $M_{2,2}$ is $\MHB(\bm{d}_1,\dots,\bm{d}_n)$. Then for any system
    $\bm{f_0} := (f_0,\dots,f_n)$,
    $V_\P(\mon{\theta},f_1,\dots,f_n) = \emptyset$ if and only if
    $M_{1,1}(\bm{f_0})$ is non-singular.
  \end{proposition}

  \begin{proof}
    Consider the determinant of $M$. As it is a multiple of the
    resultant (\Cref{sec:sparse-res}) and the resultant is a
    multihomogeneous polynomial of degree
    $\MHB(\bm{d}_1,\dots,\bm{d}_n)$ with respect to $\bm{u_0}$, we can
    write
    $\det(M) = P(\bm{u}) \cdot
    u_{0,\bm{\theta}}^{\MHB(\bm{d}_1,\dots,\bm{d}_n)} + Q(\bm{u})$,
    where $P(\bm{u}) \in \K[\bm{u}]$ does not involve the variables in
    $\bm{u}_{0}$ and $Q(\bm{u}) \in \K[\bm{u}]$ is a polynomial such
    that none of its monomials are multiple of
    $u_{0,\bm{\theta}}^{\MHB(\bm{d}_1,\dots,\bm{d}_n)}$.
    As $\Pi_\theta$ holds, $u_{0,\bm{\theta}}$ only appears in the
    diagonal of $M_{2,2}$. Consider the expansion by minors of
    $\det(M)$. If the (row) dimension of $M_{2,2}$ is
    $\MHB(\bm{d}_1,\dots,\bm{d}_n)$, then
    $P(\bm{u}) = \pm \det(M_{1,1})$.
    The polynomial $P(\bm{u})$ is a constant multiple of the cofactor
    of $u_{0,\bm{\theta}}^{\MHB(\bm{d}_1,\dots,\bm{d}_n)}$ in the
    resultant $\Res_\P(\bm{d}_0,\dots,\bm{d}_n)$.

    By construction, $Q(\bm{u})$ is a homogeneous polynomial with
    respect to the variables $\bm{u}_0$ of degree
    ${\MHB(\bm{d}_1,\dots,\bm{d}_n)}$. As
    $u_{0,\bm{\theta}}^{\MHB(\bm{d}_1,\dots,\bm{d}_n)}$ does not
    divide any monomial in $Q(\bm{u})$, each monomial involves a
    variables of $\bm{u}_0$ different to $u_{0,\theta}$.
    Hence, for any system $\bm{f_0}$, we have
    $Q(\mon{\theta},f_1,\dots,f_n) = 0$.
    By construction, the polynomial $P(\bm{u})$ does not involve any
    of the variables of $\bm{u}_{0}$. Therefore
    $\det(M_{1,1})(\bm{f_0}) =
    \det(M_{1,1})(\mon{\theta},f_1,\dots,f_n)$.
    Therefore, for any system $\bm{f_0}$,  \linebreak
    $q \cdot
    \Res_\P(\bm{d}_0,\dots,\bm{d}_n)(\mon{\theta},f_1,\dots,f_n) =
    \det(M)(\mon{\theta},f_1\dots f_n) =$ \linebreak
    $\pm \det(M_{1,1})(\mon{\theta},f_1\dots f_n) =  \pm
    \det(M_{1,1})(\bm{f_0})$.
    The determinant of $M$ is a non-zero constant multiple of the
    resultant, hence  \linebreak $\det(M_{1,1})(\bm{f_0}) \neq 0$ if and only if
    the system $(\mon{\theta},f_1,\dots,f_n)$ has no solutions over
    $\P$, i.e., $V_\P(\mon{\theta},f_1,\dots,f_n) = \emptyset$.
  \end{proof}

  If the square system $\bm{f} = (f_1,\dots,f_n)$ has no solutions
  at infinity in $\P$, that is all the coordinates of the solutions
  are not zero, then the evaluation of the solutions of $\bm{f}$ at
  any monomial in $S(\bm{d}_0)$ is not zero. Hence, for any
  $\mon{\theta} \in S(\bm{d}_0)$,
  $V_\P(\mon{\theta},f_1,\dots,f_n) = \emptyset$. By
  \Cref{thm:m11invertible}, 
  $M_{1,1}(f_0,f_1,\dots,f_n)$ is invertible.
  To avoid solutions at infinity, in the 0-dimensional
  multihomogeneous case, we perform a generic linear change of
  coordinates that preserves the multihomogeneous structure. We state
  the following corollary without proof.

  \begin{corollary}
    \label{thm:multihomogeneousOK}
    Consider a square multihomogeneous system
    $\bm{f} \in S(\bm{d}_1) \times \dots \times S(\bm{d}_n)$
    with finite  $V_\P(\bm{f})$. Choose $\theta \in \A(\bm{d}_0)$
    and let $M$ be a resultant matrix for
    $\Res_\P(d_0,\dots,d_n)$, such that $\Pi_\theta$ holds. Consider
    any $f_0 \in S(\bm{d}_0)$. Then, for a generic linear change of
    coordinates $A$, preserving the multihomogeneous structure, the
    matrix $M_{1,1}(f_0,f_1\circ A, \dots, f_n\circ~A)$~is~invertible.
  \end{corollary}

  We can use \Cref{thm:eigenvalues} to solve the 2-bilinear
  systems.

  \begin{theorem} \label{thm:chgOfCoord} 
    Assume a 2-bilinear system $f_1,\dots,f_n$ of type \linebreak
    $(n_x,n_y,n_z;r,t)$, such that $V_\P(f_1,\dots,f_n)$ is finite.
    Choose $\theta \in \A(\bm{d}_0)$ and consider the $M$ be the
    matrix of $\delta_1(\bm{F^{(2)}},\bm{m})$
    (\Cref{sec:determinantalMap}) for the ``universal'' system
    $\bm{F^{(2)}}$ rearranged with respect to the monomial
    $\mon{\theta}$. Choose $f_0 \in S(1,1,1)$.
   Then, after applying a generic linear change of coordinates
   $A$, preserving the multihomogeneous structure, the eigenvalues of
   the Schur complement of $M_{2,2}(f_0,f_1\circ A, \dots, f_n\circ A)$
   are the evaluations of $\frac{f_0}{\mon{\theta}}$ over
   $V_\P(f_1 \circ A,\dots,f_n \circ A)$.
  \end{theorem}

  \begin{proof}
    We only need to check if the Koszul resultant matrix
    has the property $\Pi_\theta$.
    The entries of our matrix are the variables of $\bm{u}$
    up to sign.
    Note that if $u_{i,\sigma} \in \bm{u}$ appears in an entry,
    then it does not appear in the other entries in the same row, or
    column. 
    Hence, we can rearrange the matrix in such a way that the
    coefficient $u_{0,\theta}$ only appears in the diagonal of $M_{2,2}$.
    As the determinant of the system is a constant multiple of the
    resultant, the dimension of $M_{2,2}$ the degree of $\bm{u_{0}}$
    in the determinant, which equals the \MHB.
  \end{proof}

  \begin{example}[Cont.]
    In the previous example (\Cref{ex:partMatrix}), we choose
    $\bm{\theta} = ((1,0),(1,0),(1,0)) \in \A(1,1,1)$ and partition
    the matrix as $\bigl[\begin{smallmatrix}M_{1,1} & M_{1,2}
        \\ M_{2,1} & M_{2,2} \end{smallmatrix} \bigr]$. If we consider
    the Schur complement, we get $\bigl[\begin{smallmatrix} 5 & -2
        \\ 4 & -1 \end{smallmatrix} \bigr]$. The characteristic
    polynomial of this matrix is $X^2-4X+3$, whose roots are
    $\frac{f_0}{\mon{\theta}}(\alpha_1) = 3$ and
    $\frac{f_0}{\mon{\theta}}(\alpha_2) = 1$.
  \end{example}
   
  \subsection{Eigenvectors for 2-bilinear systems} \label{sec:eigenvectors}

  We fix $\theta \in \A(\bm{d}_0)$.
  We consider the
  degree vector $\bm{m} = (n_y-1,-1,n_x+n_y-r+1)$ and the
  determinantal formula $M$ for the map
  $\delta_1(\bm{F^{(2)}},\bm{m})$ (\Cref{sec:determinantalMap}). We
  study the right eigenvectors of the Schur complement of $M_{2,2}$ to
  recover the coordinates of all the solutions of a 2-bilinear system
  $\bm{f}$ of type $(n_x,n_y,n_z ; r,s)$ (\Cref{sec:2-bilinear}). We
  assume that the number of different solutions is
  $\#V_\P(\bm{f}) = \MHB(\bm{f})$.

  We augment $\bm{f}$ to $\bm{f_0}$ by adding a trilinear polynomial
  $f_0$, which we specify in the sequel. We study the right
  eigenvalues of the Schur complement of $M_{2,2}(\bm{f_0})$. We
  reduce the analysis of the kernel of $\delta_1(\bm{f_0},\bm{m})$ to
  the analysis of a map in a strand of the Koszul complex of a system
  with common solutions.

  Let $\alpha = (\alpha_x,\alpha_y,\alpha_z) \in \P$, and without loss
  of generality assume that $\alpha_{t,0} \neq 0$, for
  $t \in \{x,y,z\}$.
  First, we study the kernel of $\delta_1(\bm{f_0},\bm{m})$, when the
  overdetermined system $\bm{f_0}$ \emph{has} common solutions. We relate
  this kernel to the eigenvectors, as we did in the proof of
  \cref{thm:eigenvalues}.
  For each variable $t \in \{x,y,z\}$, consider the dual form
  $$\ev^t_\alpha(d_t) := \sum_{\theta_t \in \A(d_t)}
  \frac{\bm{t}^{\theta_t}}{t_0^{d_t}}(\alpha_t) \; \bm{\partial t}^\theta
  \in S_t(d_t)^*
  $$
  for $d_t \geq 0$. If $d_t < 0$, then we take $\ev^t_\alpha(d_t) := 0$.
  
  \begin{observation}
    For each variable $t \in \{x,y,z\}$, given a polynomial
    $g_t \in S_t(\bar{d_t})$, such that $\bar{d_t} \leq d_t$, then
    operator $\star_t$, \Cref{eq:star}, acts over $g_t$ and
    $\ev^t_\alpha(d_t)$ as the evaluation of
    $\tfrac{g_t}{t_0^{\bar{d_t}}}$ at $\alpha$, that is 
    $$g_t \star_t \ev^t_\alpha(d_t) = \frac{g_t}{t_0^{\bar{d_t}}}(\alpha_t)
    \cdot \ev^t_\alpha(d_t - \bar{d_t})  .$$
  \end{observation}
  
  To simplify notation, given $f \in S(d_x,d_y,d_z)$ and
  $(\alpha_x,\alpha_y,\alpha_z) \in~\P$, we denote by
  $f(\alpha_x,\alpha_y) \in S_z(d_z)$ the partial evaluation of
  $\tfrac{f}{x_0^{d_x}y_0^{d_y}}$ at $\bm{x} = \alpha_x$ and
  $\bm{y} = \alpha_y$. This evaluation is well-defined because the
  numerator and denominator share the same degrees w.r.t.
  $\bm{x}$ and $\bm{y}$.

  \begin{lemma} \label{thm:psiAndEv}
    Consider $\bm{d} = (d_x,d_y,d_z)$, 
    $\bm{\bar{d}} = (\bar{d_x},\bar{d_y},\bar{d_z})$. Let
    $f \in S(\bm{\bar{d}})$ and $g_z \in S_z(d_z)$. If
    $d_x \geq \bar{d_x}$ and $d_y \geq \bar{d_y}$, then the map $\psi$
    (\Cref{eq:defpsi}) acts over
    $\ev^x_{\alpha}(d_x) \otimes \ev^y_\alpha(d_y) \otimes g_z$ and
    $f$, as the multiplication of $g_z$ and $f(\alpha_x,\alpha_y)$,
    that is
    $$\psi(
        \ev^x_{\alpha}(d_x) \, \otimes \, \ev^y_\alpha(d_y) \otimes  g_z, f
      ) = 
      \ev^x_\alpha(d_x - \bar{d_x})  \otimes \ev^y_\alpha(d_y - \bar{d_y}) \otimes
      \big( g_z \cdot f(\alpha_x,\alpha_y) \big).
      $$
    \end{lemma}

  Let $\omega^{(1)} := \{I : \bm{e}_I \in \!\!\!\! \bigwedge\limits_{r,s-n_z+1,0} \!\!\!\! E\}$
  and $\omega^{(2)} := \{J : \bm{e}_J \in \!\!\!\! \bigwedge\limits_{r,s-n_z,1} \!\!\!\! E\}$.
  Let $\rho_\alpha : \K^{\#\omega^{(1)}} \times \K^{\#\omega^{(2)}} \rightarrow
  L_{1,1} \oplus L_{1,2}$, \Cref{eq:K1},
  \begin{multline*} 
  \rho_\alpha(\bm{\lambda^{(1)}},\bm{\lambda^{(2)}}) :=
  \sum_{I \in \omega^{(1)}} \lambda^{(1)}_{I} \cdot \Big( \ev^x_{\alpha}(1) \otimes \ev^y_{\alpha}(r-n_y)
  \otimes 1 \otimes \bm{e}_I \Big) \\
  + \sum_{J \in \omega^{(2)}}
  \lambda^{(2)}_{J} \cdot
  \Big(
  \ev^x_{\alpha}(1) \otimes \ev^y_{\alpha}(r-n_y+1) \otimes
  1 \otimes \bm{e}_J
  \Big)
  \end{multline*}
  As $\#\omega^{(1)} + \#\omega^{(2)} = {s+1 \choose s-n_z+1}$, we
  write $\rho_\alpha : \K^{{s+1 \choose s-n_z+1}} \rightarrow K_{1}$.

  \begin{lemma} \label{thm:rhoIsKoszul} The linear map
    $\delta_1(\bm{f_0},\bm{m}) \circ \rho_\alpha : \K^{s+1 \choose s -
      n_z + 1} \rightarrow K_0 $ is equivalent to the
    $(s - n_z + 1)$-th map of the Koszul complex of the following
    system, consisting of $s+1$ linear polynomials in
    $\bm{z}$,
      \begin{align} \label{eq:fz} 
      \bm{f_z} :=
      \Big(f_0(\alpha_x,\alpha_y),f_{r+1}(\alpha_x,\alpha_y),\dots,f_{n}(\alpha_x,\alpha_y)\Big),
      \end{align}
      restricted to its 0-graded part, i.e. the strand of the Koszul
      complex such that its $(s - n_z + 1)$-th module is isomorphic to
      $\K^{s + 1 \choose s - n_z + 1}$.
    \end{lemma}

    \noindent
    If $\bm{f_0}$ has a solution
    $(\alpha_x,\alpha_y,\alpha_z) \in V_\P(\bm{f_0})$, then,
    $\alpha_z$ is a solution of the linear system $\bm{f_z}$,
    that is $\alpha_z \in V_\P(\bm{f_z})$. As $\bm{f_z}$
    is an overdetermined system, the Koszul complex $\bm{f_z}$ is not
    exact \cite[Thm.~XXI.4.6]{lang2002algebra}.
    %

    \begin{lemma} \label{thm:lambdaInKernel}
      Let $\bm{f_0}$ be an overdetermined 2-bilinear
      system. If $\alpha \in V_\P(\bm{f_0})$, then there is a non-zero
      $\bm{{\widehat{\lambda}_\alpha}} \in \K^{s+1 \choose s - n_z + 1}$ such
      that $\delta_1(\bm{f_0},\bm{m}) \circ
      \rho_\alpha(\bm{\widehat{\lambda}_\alpha}) = 0$.
    \end{lemma}

    \begin{proof}
      Following \Cref{thm:rhoIsKoszul}, if we compose
      $\delta_1(\bm{f},\bm{m})$ and $\rho_\alpha$, then we obtain a
      map which is similar to the 0-graded part of the
      $(s - n_z + 1)$-th map of the Koszul complex of the $s+1$ linear
      polynomials in $\bm{z}$, $\bm{f_z}$, \Cref{eq:fz}.
      As the linear system $\bm{f_z}$ has a solution $\alpha_z$, at
      most $n_z$ of its polynomials are linearly independent. Hence, the
      Koszul complex of $\bm{f_z}$ is isomorphic to a Koszul complex
      {$K(\widetilde{f}_1, \dots, \widetilde{f}_{n_z}, 0, \dots, 0)$} of
      a system of $s+1$ linear polynomials, where $(s+1-n_z)$ of them
      are equal to zero \cite[Lem.~XXI.4.2]{lang2002algebra}. The
      $(s+1-n_z)$-th map of
      {$K(\widetilde{f}_1, \dots, \widetilde{f}_{n_z}, 0, \dots, 0)$}
      maps $e_{n_z+1}\!\wedge\!\dots\!\wedge\!{e_{s+1-n_z}}$ to zero.
      Hence, its 0-graded part has a non-trivial kernel, and so there
      is a non-zero
      $\bm{\widehat{\lambda}_\alpha} \in \K^{s \choose s - n_z + 1}$ such that
      $\delta_1(\bm{f},\bm{m}) \circ \rho_\alpha(\bm{\widehat{\lambda}_\alpha})
      = 0$.
    \end{proof}
    
    \begin{theorem} \label{thm:eigenvectors}
      Let  $\bm{f} = (f_1,\dots,f_n)$ be a  square 2-bilinear system
      of type $(n_x,n_y,n_z;r,s)$, such that it has
      ${r \choose n_y} \cdot {s \choose n_z}$ different solutions over
      $\P$. Consider $\theta \in \A(1,1,1)$ such that
    $$\Res_\P^{(2)}(\mon{\theta},f_1,\dots,f_n) \neq 0$$ and $f_0 \in
    S(1,1,1)$ such that $\frac{f_0}{\mon{\theta}}$ separates the
    elements in $V_\P(\bm{f})$. Let $\bm{m} := (n_y-1,-1,n_x+n_y-r+1)$
    and $M \in \K[\bm{u}]^{\mathcal{K} \times \mathcal{K}}$ related to
    $\delta_1(\bm{F^{(2)}}, \bm{m})$ for the overdetermined 2-bilinear
    ``universal'' system (\Cref{thm:goodDegreeVectors}). Then, the
    Schur complement of $M_{2,2}(\bm{f_0})$ is diagonalizable, each
    eigenvalue is $\frac{f_0}{\mon{\theta}}(\alpha)$, for
    $\alpha \in V_\P(f_1,\dots,f_n)$, and we can extend the
    eigenvector $\bar{v}_\alpha$ related to $\alpha$ to
    $v_\alpha :=
    \bigl[\begin{smallmatrix} M_{1,1}^{-1} \cdot M_{2,1} \\
        I \end{smallmatrix}\bigr](\bm{f_0}) \cdot
    \bar{v}_\alpha$ such that $v_\alpha$ is the element
    $\rho_\alpha(\bm{\widehat{\lambda}_\alpha})$, for some
    $\bm{\widehat{\lambda}_\alpha} \in \K^{s+1 \choose s - n_z + 1}$.
    \end{theorem}

    \begin{proof}
      By \Cref{thm:separateRootsImpliesDiagonalizable}, the Schur
      complex of $M_{2,2}(\bm{f_0})$ is diagonalizable and every
      eigenvalues is different.
      For each $\alpha \in V_\P(\bm{f})$, consider the eigenvalue
      $\frac{f_0}{\mon{\theta}}(\alpha)$, related eigenvector
      $\bar{v}_\alpha$, and the system
      $\bm{g_\alpha} := (f_0 -
      \frac{f_0}{\mon{\theta}}(\alpha),f_1,\dots,f_n)$.
      By \Cref{thm:lambdaInKernel}, there is a
      $\bm{\lambda_\alpha} \in \K$ such that
      $\delta_1(\bm{g_\alpha},\bm{m}) \circ \rho(\bm{\lambda_\alpha}) =
      0$.
      Hence, there is a $w_\alpha$, representing
      $\rho(\bm{\lambda_\alpha}) = 0$, in the kernel of $M(\bm{g_\alpha})$.
      Following the proof of \Cref{thm:eigenvalues}, each element in
      the kernel of the Schur complement of $M_{2,2}(\bm{g_\alpha})$ is
      related to an eigenvector of the Schur complement of
      $M_{2,2}(\bm{f_0})$ with corresponding eigenvalue
      $\frac{f_0}{\mon{\theta}}(\alpha)$.
      As for each eigenvalue we have only one eigenvector, then the
      dimension of this kernel is $1$.
      Hence, the truncation of $w_\alpha$,
      $\bar{w}_\alpha := (0 | I) \cdot w_\alpha$, is a
      multiple of $\bar{v}_\alpha$, where $0$ is the zero
      matrix of appropriate dimension.

      As $M_{1,1}(\bm{g_\alpha})$ is invertible and
      $M(\bm{g_\alpha}) \cdot w_\alpha = 0$, it holds that \linebreak
      $\bigl[\begin{smallmatrix} M_{1,1}^{-1} \cdot M_{2,1} \\
        I \end{smallmatrix}\bigr](\bm{g_\alpha}) \bar{w}_\alpha =~
      w_\alpha$. As $\bigl[\begin{smallmatrix} M_{1,1}^{-1} \cdot M_{2,1} \\
        I \end{smallmatrix}\bigr](\bm{g_\alpha})$ does not involve
      $u_{0,\theta}$, then
      $\bigl[\begin{smallmatrix} M_{1,1}^{-1} \cdot M_{2,1} \\
        I \end{smallmatrix}\bigr](\bm{g_\alpha}) = 
      \bigl[\begin{smallmatrix} M_{1,1}^{-1} \cdot M_{2,1} \\
        I \end{smallmatrix}\bigr](\bm{f_0})$.
      Therefore, we conclude that, as $\bar{v}_\alpha$ is a multiple
      of $\bar{w}_\alpha$, then
      $v_\alpha = \bigl[\begin{smallmatrix} M_{1,1}^{-1} \cdot M_{2,1} \\
        I \end{smallmatrix}\bigr](\bm{f_0}) \cdot \bar{v}_\alpha$ is a
      multiple of $w_\alpha$.
    \end{proof}

    In the following example we use \Cref{thm:eigenvectors} to recover $\alpha_2$.

    \begin{example}[Cont.]

      The eigenvalue of $\frac{f_0}{\mon{\theta}}(\alpha_2) = 1$ is
      $\bar{v}_{\alpha_2} := (1,2)^{\top}$.
      By extending $\bar{v}_{\alpha_2}$, we get $$v_{\alpha_2} :=
      \bigl[\begin{smallmatrix} M_{1,1}^{-1} \cdot M_{2,1} \\
        I \end{smallmatrix}\bigr](\bm{f_0}) \cdot
      \bigl(\begin{smallmatrix} 1 \\ 2 \end{smallmatrix}\bigr) =
      (4, 3, 12, 1, 2, 3, 6, 6, 1, 2)^{\top}$$ which represents $\rho_{\alpha_2}(1,1) =$
    { 
      \begin{multline*}
       \left( \bm{\partial x}^{(1,0)} +3\,\bm{\partial x}^{(0,1)} \right)
        \otimes
        \left( \bm{\partial y}^{(2,0)} +2\,\bm{\partial y}^{(1,1)}
          +4\,\bm{\partial y}^{(0,2)}\right)
        \otimes 
        1 \otimes \bm{e}_{\{0,1,2\}} \\
        +
        \left( \bm{\partial x}^{(1,0)} +3\,\bm{\partial x}^{(0,1)}
        \right)
        \otimes
        \left( \bm{\partial y}^{(1,0)} +2\,\bm{\partial y}^{(0,1)}
        \right)
        \otimes
        1 \otimes \bm{e}_{\{1,2,3\}}
      \end{multline*}  }

\noindent
Hence, {
  $\ev^x_{\alpha_2}(1) = \left( 1\, \bm{\partial x}^{(1,0)}
    +3\,\bm{\partial x}^{(0,1)} \right) $}, and so 
 $\alpha_{2,x} = (1:3) \in \Pr^1$. Also,
{$\ev^y_{\alpha_2}(1) = \left( 1\,\bm{\partial y}^{(1,0)}
    +2\,\bm{\partial y}^{(0,1)} \right) $}, and then
$\alpha_{2,y} = (1:2) \in \Pr^1$.
We note that {
$\ev^y_{\alpha_2}(2) =
        \left(1 \cdot 1 \cdot \bm{\partial y}^{(2,0)} + 1 \cdot 2 \cdot \bm{\partial y}^{(1,1)}
          +2 \cdot 2 \cdot \bm{\partial y}^{(0,2)}\right)$}.

    We can recover $\alpha_{2,z}$ as the solution of
    $\bm{f}(\alpha_{2,x},\alpha_{2,y},\bm{z}) = 0$,
    {
    \begin{align*}
      \begin{cases}
        \; f_1(\alpha_{2,x},\alpha_{2,y},\bm{z}) = 0 \\
        \; f_2(\alpha_{2,x},\alpha_{2,y},\bm{z}) = 0 \\
        \; f_3(\alpha_{2,x},\alpha_{2,y},\bm{z}) =  -9 \, z_0 + 3 \, z_1
      \end{cases}
    \end{align*}
    }
    Hence, $\alpha_{2,z} = (1 : 3) \in \Pr^1$ and so $\alpha_2 = (1\!:\!3 \;;\; 1\!:\!2 \;;\; 1\!:\!3) \in \P$.
  \end{example}
    
  \section{Size of  matrices and FG\lowercase{b}}
  \label{sec:complexity}
  As there are no tight bounds for the complexity of Gr\"{o}bner basis
  algorithms for solving 2-bilinear systems, we compare against our
  algorithms experimentally in \Cref{tab:ratio}.  We consider the state-of-the-art
  Gr\"{o}bner basis implementation, FGb \cite{FGb}. For each set of
  parameters, we consider a random square $2$-bilinear system and we
  dehomogenize the system to compute its Gr\"{o}bner basis. We compared
  the ratio between the size of the maximal matrix appearing in the
  Gr\"{o}bner basis computation and the size of our Koszul resultant
  matrix, for all the cases $n \leq 15$. For reasons of space we only
  present some indicative examples for $n = 12$. The rest of the cases
  can be found in
  \url{http://www-polsys.lip6.fr/~bender/2bilinear/}. The results
  are promising and motivate the study of the structure Koszul
  resultant matrix to develop algorithms for faster linear
  algebra with such matrices.

  \begin{table}[h] 
    \caption{\rm Matrix sizes and ratios of Koszul matrix and FGb.}
    \label{tab:ratio}
    \center
    \begin{tabular}{ | c | c | c | c | c || c | c | c |}
      \hline
      $n_x$ & $n_y$ & $n_z$ & $r$ & $s$ & Size $\delta_1$& Size FGb     & Ratio \\ \hline
      $2$  &  $6$   & $4$  & $7$  & $5$ & $630 \times 630$   & $1769 \times 1158$ & $5.1 \sim$ \\
      $10$ &  $1$   & $1$  & $10$ & $2$ & $352 \times 352$   & $709 \times 422$   & $2.4 \sim$ \\
      $5$  &  $5$   & $2$  & $9$  & $3$ & $6804 \times 6804$ & $8941 \times 8390$ & $1.6 \sim$ \\
      $4$  &  $4$   & $4$  & $6$  & $6$ & $4125 \times 4125$ & $5436 \times 4262$ & $1.3 \sim$ \\
      $5$  &  $5$   & $2$  & $6$  & $6$ & $2106 \times 2106$ & $2007 \times 1164$ & $1/1.9 \sim$ \\
      $6$  &  $3$   & $3$  & $6$  & $6$ & $7000 \times 7000$ & $4708 \times 3801$ & $1/2.7 \sim$ \\
      $6$  &  $4$   & $2$  & $5$  & $7$ & $2450 \times 2450$ & $1773 \times 1125$ & $1/3 \sim$ \\ 
      \hline
    \end{tabular}
  \end{table}

    \paragraph*{Acknowledgments:}
    We thank Laurent Bus\'e and Carlos D'Andrea for helpful
    discussions and references, and the
    anonymous reviewers for the comments and suggestions.
    The authors are partially supported by ANR JCJC GALOP
    (ANR-17-CE40-0009) and the PGMO grant GAMMA.

  \bibliographystyle{alpha}
\bibliography{bibliography}

\newpage
\appendix
\section*{Appendix}

  \begin{proof}[{Proof of \Cref{thm:psiAndEv}}]
    Consider $f = \sum_\sigma c_\sigma
    \bm{x}^{\sigma_x}\bm{y}^{\sigma_y}\bm{z}^{\sigma_z}$. As $\psi$ is
    a bilinear map and the tensor product is multilinear, it is enough
    to prove this lemma only for the monomials
    $\bm{x}^{\sigma_x}\bm{y}^{\sigma_y}\bm{z}^{\sigma_z} \in
    S(\bm{\bar{d}})$.

    , $\psi(\ev^x_{\alpha}(d_x) \otimes
    \ev^y_\alpha(d_y) \otimes g, f) = \sum_\sigma c_\sigma
    \psi({\ev^x_{\alpha}(d_x) \otimes \ev^y_\alpha(d_y) \otimes
    g}, \\ \bm{x}^{\sigma_x}\bm{y}^{\sigma_y}\bm{z}^{\sigma_z})$. For that reason, we study the monomial case,
    
    {
      \begin{multline*}
        \psi(
        \ev^x_{\alpha}(d_x) \otimes \ev^y_{\alpha}(d_y) \otimes g_z,
        \bm{x}^{\sigma_x} \otimes \bm{y}^{\sigma_y} \otimes \bm{z}^{\sigma_z}
        ) = \\
        \Big(
        \bm{x}^{\sigma_x} \star_x \ev^x_\alpha(d_x)
        \Big)
        \otimes
        \Big(
        \bm{y}^{\sigma_y} \star_y \ev^y_\alpha(d_y)
        \Big)
        \otimes
        \Big(
        g_z \cdot \bm{z}^{\sigma_z}
        \Big)
        = \\
        \Big(
        \frac{\bm{x}^{\sigma_x}}{x_0^{\bar{d_x}}}(\alpha_x) \, \ev^x_\alpha(d_x - \bar{d_x})
        \Big)
        \otimes
        \Big(
        \frac{\bm{y}^{\sigma_y}}{y_0^{\bar{d_y}}}(\alpha_y) \, \ev^y_\alpha(d_y-\bar{d_y})
        \Big)
        \otimes
        \Big(
        g_z \cdot \bm{z}^{\sigma_z}
        \Big)
        = \\
        \Big(\ev^x_\alpha(d_x - \bar{d_x})\Big) \otimes \Big(\ev^y_\alpha(d_y - \bar{d_y})\Big) \otimes
        \Big(
        g_z \cdot \frac{\bm{x}^{\sigma_x}}{x_0^{\bar{d_x}}}(\alpha_x)
        \frac{\bm{y}^{\sigma_y}}{y_0^{\bar{d_y}}}(\alpha_y)
        \cdot \bm{z}^{\sigma_z}
        \Big)
      \end{multline*}
    }
    Then, we have
    \begin{multline*} 
      \psi(\ev^x_{\alpha}(d_x) \otimes
      \ev^y_\alpha(d_y) \otimes g, f) = \\
      \sum_\sigma c_\sigma (\ev^x_\alpha(d_x - \bar{d_x}) \otimes \ev^y_\alpha(d_y - \bar{d_y}) \otimes g_z \cdot
      \frac{\bm{x}^{\sigma_x}}{x_0^{\bar{d_x}}}(\alpha_x) \cdot
      \frac{\bm{y}^{\sigma_y}}{y_0^{\bar{d_y}}}(\alpha_y) \cdot \bm{z}^{\sigma_z}) \\
      \ev^x_\alpha(d_x - \bar{d_x}) \otimes \ev^y_\alpha(d_y - \bar{d_y}) \otimes
      g_z \cdot \sum_\sigma c_\sigma
      \frac{\bm{x}^{\sigma_x}}{x_0^{\bar{d_x}}}(\alpha_x) \cdot
      \frac{\bm{y}^{\sigma_y}}{y_0^{\bar{d_y}}}(\alpha_y) \cdot \bm{z}^{\sigma_z} \\
      \ev^x_\alpha(d_x - \bar{d_x})  \otimes \ev^y_\alpha(d_y - \bar{d_y}) \otimes
      g_z \cdot f(\alpha_x,\alpha_y)
    \end{multline*}
  \end{proof}

  \begin{proof}[{Proof of \Cref{thm:rhoIsKoszul}}]

    We split the map $\rho$ as
    $\rho(\bm{\lambda^{(1)}}, \bm{\lambda^{(2)}}) :=
    \rho_\alpha^{(1)}(\bm{\lambda^{(1)}}) +
    \rho_\alpha^{(2)}(\bm{\lambda^{(2)}})$, where
    $\rho_\alpha^{(1)} : \K^{\# \omega^{(1)}} 
    \rightarrow L_{1,1}$, Eq.~\eqref{eq:K1}, such that,
    $$\rho_\alpha^{(1)}(\bm{\lambda^{(1)}}) :=
    \sum\nolimits_{I \in \omega^{(1)}} \Big( \ev^x_{\alpha}(1) \otimes \ev^y_{\alpha}(r-n_y)
    \otimes \lambda^{(1)}_{I} \otimes \bm{e}_I \Big),$$
    and $\rho_\alpha^{(2)} : \K^{\#\omega^{(2)}} 
    \rightarrow L_{1,2}$, Eq.~\eqref{eq:K1}, such that
    $$\rho_\alpha^{(2)}(\bm{\lambda^{(2)}}) :=
    \sum_{J \in \omega^{(2)}} \Big(
    \ev^x_{\alpha}(1) \otimes \ev^y_{\alpha}(r-n_y+1) \otimes
    \lambda^{(2)}_{J} \otimes \bm{e}_J
    \Big).$$

    Both maps are injective.

    As
    $\delta_1(\bm{f_0},\bm{m}) \circ \rho_\alpha =
    \delta_1(\bm{f_0},\bm{m}) \circ \rho_\alpha^{(1)} +
    \delta_1(\bm{f_0},\bm{m}) \circ \rho_\alpha^{(2)}$, we study
    $\delta_1(\bm{f_0},\bm{m}) \circ \rho_\alpha^{(1)}$ and
    $\delta_1(\bm{f_0},\bm{m}) \circ \rho_\alpha^{(2)}$ separately.

    Following the definition of $\delta_1$
    (\Cref{sec:determinantalMap}) we have
      {
      \begin{multline*}
      \delta_1(\bm{f_0},\bm{m}) \circ \rho_\alpha^{(1)} = \\
      \sum_{I \in \omega^{(1)}} \lambda^{(1)}_{I} \delta_1(\bm{f_0},\bm{m})\Big(
      \ev^x_{\alpha}(1) \otimes \ev^y_{\alpha}(r-n_y) \otimes
      1 \otimes \bm{e}_I
      \Big) = \\
      \sum_{I \in \omega^{(1)}} \lambda^{(1)}_I
      \Big(
      \sum_{i = 1}^r (-1)^{i-1}
      \psi(
        \ev^x_{\alpha}(1) \,  \otimes \, \ev^y_\alpha(r-n_y) \otimes  1, f_{I_i}
      ) \otimes \bm{e}_{I \setminus \{I_i\}}  + \\
      \sum_{i = r+1}^{n_x+n_y+1} (-1)^{i-1}
      \psi(
        \ev^x_{\alpha}(1) \,  \otimes \, \ev^y_\alpha(r-n_y) \otimes  1, f_{I_i}
      ) \otimes \bm{e}_{I \setminus \{I_i\}}
      \Big).
      \end{multline*}
      }
      By \Cref{thm:psiAndEv} we have,
      {
      \begin{multline*}
      \delta_1(\bm{f_0},\bm{m}) \circ \rho_\alpha^{(1)} = \\       
      \sum_{I \in \omega^{(1)}}  \lambda^{(1)}_I
      \Big(
      \sum_{i = 1}^r (-1)^{i-1}
      \ev^x_\alpha(0)  \otimes \ev^y_\alpha(r-n_y-1) \otimes
      f_{I_i}(\alpha_x,\alpha_y)
      \otimes \bm{e}_{I \setminus \{I_i\}}  + \\
      \sum_{i = r+1}^{n_x+n_y+1} (-1)^{i-1}
      \ev^x_\alpha(0)  \otimes \ev^y_\alpha(r-n_y) \otimes
      f_{I_i}(\alpha_x,\alpha_y)
      \otimes \bm{e}_{I \setminus \{I_i\}}
      \Big).
    \end{multline*}
  }

  For $i \leq r$, $f_{I_i} \in S(1,1,0)$. Hence
  $f_{I_i}(\alpha_x,\alpha_y) = f_{I_i}(\alpha) = 0$.
      {
      \begin{align*}
      & \delta_1(\bm{f_0},\bm{m}) \circ \rho_\alpha^{(1)} = \\ &
      \sum_{I \in \omega^{(1)}} \lambda^{(1)}_I
      \sum_{i = r+1}^{n_x+n_y+1} (-1)^{i-1}
      \ev^x_\alpha(0)  \otimes \ev^y_\alpha(r-n_y) \otimes
      f_{I_i}(\alpha_x,\alpha_y)
      \otimes \bm{e}_{I \setminus \{I_i\}} = \\ &
      \ev^x_\alpha(0)  \otimes \ev^y_\alpha(r-n_y) \otimes
      \big( \sum_{I \in \omega^{(1)}}
      \sum_{i = r+1}^{n_x+n_y+1} (-1)^{i-1}  \lambda^{(1)}_I f_{I_i}(\alpha_x,\alpha_y)
      \otimes \bm{e}_{I \setminus \{I_i\}} \big).
      \end{align*}
      }
      We conclude that the image of $\delta_1(\bm{f_0},\bm{m}) \circ
      \rho_\alpha^{(1)}$ belongs to $L_{0,2}$.

      Now consider $\delta_1(\bm{f_0},\bm{m}) \circ
      \rho_\alpha^{(2)}$. Following a similar procedure, we deduce

      {
      \begin{multline*}
        \delta_1(\bm{f_0},\bm{m}) \circ \rho_\alpha^{(2)} = \\
        \ev^x_\alpha(0) \otimes \ev^y_\alpha(r-n_y) \otimes
        \sum_{I \in \omega^{(2)}}
        \Big( \lambda^{(2)}_I f_0(\alpha_x,\alpha_y) \otimes \bm{e}_{I
          - \{0\}} \Big) + \\
        \ev^x_\alpha(0) \otimes \ev^y_\alpha(r-n_y-1) \otimes
        \sum_{I \in \omega^{(2)}} \sum_{i = 2}^{r+1}
        \Big( (-1)^{i-1} \lambda^{(2)}_I f_{I_i}(\alpha_x,\alpha_y) \otimes \bm{e}_{I - \{I_i\}} \Big) + \\
        \ev^x_\alpha(0) \otimes \ev^y_\alpha(r-n_y+1) \otimes
        \sum_{I \in \omega^{(2)}} \sum_{i = r+1}^{n_x+n_y+1}
        \Big( (-1)^{i-1} \lambda^{(2)}_I f_{I_i}(\alpha_x,\alpha_y) \otimes \bm{e}_{I - \{I_i\}} \Big)
      \end{multline*}
     }

     For $1 \leq i \leq r+1$, $f_{I_i} \in S(1,1,0)$, so
     $f_{I_i}(\alpha_x,\alpha_y) = f_{I_i}(\alpha) = 0$. Hence,
      {
      \begin{multline*}
        \delta_1(\bm{f_0},\bm{m}) \circ \rho_\alpha^{(2)} = \\
        \ev^x_\alpha(0) \otimes \ev^y_\alpha(r-n_y) \otimes
        \sum_{I \in \omega^{(2)}}
        \Big( \lambda^{(2)}_I f_0(\alpha_x,\alpha_y) \otimes \bm{e}_{I
          - \{0\}} \Big) + \\
        \ev^x_\alpha(0) \otimes \ev^y_\alpha(r-n_y+1) \otimes
        \sum_{I \in \omega^{(2)}} \sum_{i = r+1}^{n_x+n_y+1}
        \Big( (-1)^{i-1} \lambda^{(2)}_I f_{I_i}(\alpha_x,\alpha_y) \otimes \bm{e}_{I - \{I_i\}} \Big)
      \end{multline*}
    } 
    Therefore, the image of
    $\delta_1(\bm{f_0},\bm{m}) \circ \rho_\alpha^{(2)}$ belongs to
    $L_{0,2} \oplus L_{0,4}$.

    We can rewrite
    $\delta_1(\bm{f_0},\bm{m}) \circ \rho_\alpha : \K^{s+1 \choose s
      - n_z + 1} \rightarrow L_{0,2} \oplus L_{0,4}$ as
      \begin{multline*} 
          (\delta_1(\bm{f_0},\bm{m}) \circ \rho_\alpha)(\bm{\lambda}) = 
          \ev^x_{\alpha}(0) \otimes \ev^y_{\alpha}(r-n_y) \otimes P_1(\bm{\lambda}) + \\
          \ev^x_{\alpha}(0) \otimes \ev^y_{\alpha}(r-n_y+1) \otimes (-1)^r P_2(\bm{\lambda})
      \end{multline*}
      where
      {
      \begin{multline*}
        P_1(\bm{\lambda}) := \sum_{I \subset \omega^{(2)}}
        \lambda_{I} f_0(\alpha_x,\alpha_y) \otimes \bm{e}_{I \setminus \{0\}}  + \\ 
        \sum_{J \subset \omega^{(1)}} \sum_{j = 1}^{s-n_z+1}
        (-1)^{j-1} \lambda_{J} f_{J_j}(\alpha_x,\alpha_y) 
        \otimes \bm{e}_{J \setminus \{J_j\}}
      \end{multline*}

      $$P_2(\bm{\lambda}) := \sum_{I \subset \omega^{(2)}}
        \sum_{j = 2}^{s-n_z} (-1)^{r+j-1}
        \lambda_{I} f_{I_j}(\alpha_x,\alpha_y)
        \otimes \bm{e}_{I \setminus \{I_j\}}$$
    }
    
      We observe that the intersection between the image of $P_1$ and
      $- P_2$ is trivial, because $\Im(P_1) \in S_z(1) \otimes
      \bigwedge_{r,s-n_z,0} E$ and $\Im(P_2) \in S_z(1) \otimes
      \bigwedge_{r,s-n_z-1,1} E$. Hence, $P_1 + P_2$ vanishes if and
      only if $P_1$ and $P_2$ vanish. Hence, $\delta_1 \circ
      \rho_\alpha$ is equivalent to the map $\bm{\lambda} \mapsto
      P_1(\bm{\lambda}) + P_2(\bm{\lambda})$. Note that, for all $I
      \in \omega^{(1)} \cup \omega^{(2)}$, $\{1,\dots,r\} \subset
      I$. Therefore, if we expand this map we conclude that it is
      equivalent to the 0-graded part of the $(s-n_z+1)$-th map of the
      Koszul complex of the linear system $\bm{f_z}$.
      \begin{multline*}
        P_1(\bm{\lambda}) + P_2(\bm{\lambda}) =  
        \sum_{\substack{J \subset \{0,r+1,\dots,n\} \\ \#J = s-n_z+1 }}
        \sum_{j = 1}^{s-n_z+1}
        (-1)^{j-1} \lambda_{J} f_{J_j}(\alpha_x,\alpha_y)
        \otimes \bm{e}_{\{1\dots r\} \cup J \setminus \{J_j\}} \qedhere
      \end{multline*}
    \end{proof}
\end{document}